\documentclass{article}
\usepackage{zack}

\usepackage{blkarray, bigstrut}
\usepackage{xparse}

\newcommand{\textQMCdProb}{\textsc{Quantum Max-\(d\)-Cut}\xspace}
\newcommand{\QMCdProb}{\textsc{QMax-\(d\)-Cut}}
\newcommand{\textQMCProb}{\textsc{Quantum Max-Cut}\xspace}

\newcommand{\textMCdProb}{\textsc{Max-\(d\)-Cut}\xspace}

\newcommand{\textMCProb}{\textsc{Max-Cut}\xspace}

\title{Approximation Algorithms for \texorpdfstring{\textQMCdProb}{Quantum Max-d-Cut}}

\let\truehypersetup\hypersetup
\renewcommand\hypersetup[1]{}
\usepackage{bigfoot}
\let\hypersetup\truehypersetup

\DeclareNewFootnote{AAffil}[arabic]
\DeclareNewFootnote{ANote}[fnsymbol]

\usepackage{etoolbox}
\makeatletter
\patchcmd\maketitle{\def\@makefnmark{\rlap{\@textsuperscript{\normalfont\@thefnmark}}}}{}{}{}
\makeatother

\makeatletter
\def\thanksAAffil#1{%
  \footnotemarkAAffil\protected@xdef\@thanks{\@thanks%
        \protect\footnotetextAAffil[\the \c@footnoteAAffil]{#1}}%
}
\def\thanksANote#1{%
  \footnotemarkANote%
  \protected@xdef\@thanks{\@thanks%
        \protect\footnotetextANote[\the \c@footnoteANote]{#1}}%
}
\makeatother

\author{%
    Charlie Carlson%
    \thanksAAffil{University of California, Santa Barbara}$^{,}$\thanksANote{charlieannecarlson@ucsb.edu}%
    \and%
    Zackary Jorquera%
    \thanksAAffil{University of California, Santa Cruz}$^{,}$\thanksANote{\{zjorquer,akolla,skordono,swayland\}@ucsc.edu}%
    \and%
    Alexandra Kolla%
    \footnotemarkAAffil[2]$^{,}$\footnotemarkANote[2]%
    \and%
    Steven Kordonowy\footnotemarkAAffil[2]$^{,}$\footnotemarkANote[2]%
    \and%
    Stuart Wayland\footnotemarkAAffil[2]$^{,}$\footnotemarkANote[2]%
}
\date{}

\begin{document}

\maketitle

\begin{abstract}
	We initiate the algorithmic study of the \textQMCdProb problem, a quantum generalization of the well-known \textMCdProb problem. The \textQMCdProb problem involves finding a quantum state that maximizes the expected energy associated with the projector onto the antisymmetric subspace of two, \(d\)-dimensional qudits over all local interactions. Equivalently, this problem is physically motivated by the \(SU(d)\)-Heisenberg model, a spin glass model that generalized the well-known Heisenberg model over qudits. We develop a polynomial-time randomized approximation algorithm that finds product-state solutions of mixed states with bounded purity that achieve non-trivial performance guarantees. Moreover, we prove the tightness of our analysis by presenting an algorithmic gap instance for \textQMCdProb with \(d \geq 3\).
\end{abstract}

\section{Introduction}

The quantum Heisenberg model is a family of spin glass Hamiltonians defined by nearest-neighbor interactions \cite{heisenberg_zur_1928}. This model, especially the antiferromagnetic variant, is well-studied in condensed matter physics \cite{heisenberg_zur_1928,bloch_zur_1930,bethe_zur_1931,lieb_ordering_1962,arovas_functional_1988,fields_renormalization_2008} and has recently gained attention in computer science since it can be seen as a quantum generalization of the \textMCProb problem. The quantum Heisenberg model is also commonly referred to as \textQMCProb in the literature \cite{gharibian_almost_2019,anshu_beyond_2020,parekh_application_2021,parekh_beating_2021,hwang_unique_2022,king_improved_2022,lee_optimizing_2022,parekh_optimal_2022,watts_relaxations_2023,takahashi_su2-symmetric_2023}. In this work, we consider a generalization of the quantum Heisenberg model, known as the $SU(d)$-Heisenberg model \cite{kawashima_ground_2007,beach_su_2009,lou_antiferromagnetic_2009}, or \textQMCdProb, that deals with interactions of spins with local Hilbert space of dimension $d$. Similarly to \textQMCProb, \textQMCdProb can be seen as the quantum generalization of \textMCdProb. Additionally, this model is known to be universal and \cc{QMA}-hard to optimize \cite{piddock_universal_2021}.

There has been a large body of literature recently that focuses on finding classical approximation algorithms for \textQMCProb \cite{gharibian_almost_2019,anshu_beyond_2020,parekh_application_2021,king_improved_2022,lee_optimizing_2022,parekh_optimal_2022,takahashi_su2-symmetric_2023}, while not much is known for \textQMCdProb. In this paper, we initiate a systematic algorithmic study of \textQMCdProb on general graphs and provide algorithmic results summarized in \crefrange{thm_1}{thm_2} in the ``Our Results" section below. 

Our work on the \textQMCdProb problem can be seen as a quantum extension to a long line of research on Generalized Constraint Satisfaction Problems (GCSPs), which are central to classical computer science. In a GCSP, there is a set of input variables $\mathcal{V} = \{x_1, \dots, x_n\}$ taking values over a finite domain $[d] = \{1,\dots,d\}$ and a distribution, \(\mathcal{P}\), over a set of payoff functions $\{P_1, \dots, P_m\}$ where each payoff function $P_j: [d]^k \ra [-1,1]$ specifies how a subset of the variables interact with each another. The objective is to maximize the expected payoff. It is common to fix the arity, the maximum number of inputs to each payoff function, to some value \(k\) (often denoted with \(k\)-GCSP). Finding the optimal assignment to a GCSP is usually computationally intractable, so one relaxes to the optimization version: find an assignment that achieves the largest expected payoff possible. An algorithm is an $\alpha$-approximation for the GCSP if it can output a solution that is guaranteed to achieve a solution with value $\geq \alpha \text{OPT}$, where OPT is the maximum expected payoff. 
Every GCSP admits the canonical $\alpha$-approximation algorithm using a relaxation to a semidefinite program (SDP) followed by a ``rounding" to a solution in the original space \cite{raghavendra_optimal_2008,raghavendra_approximating_2009}. The goal is to use the SDP solution's value as a way to bound the performance of the rounded solution. Assuming the unique games conjecture (UGC), it is NP-hard to improve upon this algorithm in general for GCSPs \cite{raghavendra_optimal_2008}.

A popular class of GCSPs can be phrased as coloring problems on graphs. Given $G = (V,E)$, the goal is to find an assignment $\tau : V \rightarrow [d]$ that maximizes the size of the set $\{(u,v) \in E\ |\ \tau(u) \neq \tau(v)\}$. The simplest non-trivial case of $d=2$ is called \textMCProb, a standard NP-complete problem \cite{karp_reducibility_1972}. The gold standard algorithm for \textMCProb was provided by Goemans and Williamson and is a $0.878$-approximation algorithm that relies on SDPs and random hyperplane rounding \cite{goemans_improved_1995}. This algorithm is the basis for the aforementioned canonical GCSP algorithm, so improving on this bound is NP-hard (assuming UGC) \cite{khot_optimal_2007,mossel_noise_2010}. Even without UGC, it is known that it is NP-hard to do better than a $16/17$-approximation \cite{hastad_optimal_2001,trevisan_gadgets_2000}. An extension to \textMCProb is \textMCdProb, in which the domain is increased to some $d \geq 2$. A simple randomized algorithm for \textMCdProb results in a $(1-1/d)$-approximation. Frieze and Jerrum extend \cite{goemans_improved_1995} to an algorithm that additively improves on the randomized bound by $\Theta\left(\frac{\ln{d}}{d^2}\right)$ \cite{frieze_improved_1997}. Other attempts at generalizing Goemans and Williamson to $d > 2$ also produce similar approximation guarantees \cite{raghavendra_optimal_2008}.

The quantum analog of a $k$-GCSP is the $k$-local Hamiltonian problem ($k$-LHP): the variables are qudits, and payoffs are represented by local observables, $\{h_1, \dots, h_m\}$, each acting non-trivially over $\leq k$ qudits. The $\{h_\alpha\}$ terms are known as the local Hamiltonians and it is common to assume they are all PSD \cite{gharibian_approximation_2012}. The problem Hamiltonian is given by $H = \E_\alpha h_\alpha$ and the task is finding the maximum eigenvalue $\lambda_{\text{max}}(H)$ \cite{gharibian_approximation_2012}.\footnote{Often, $k$-LHPs concern themselves with finding the ground state or $\lambda_{\text{min}}(H)$. However, as was done in \cite{gharibian_approximation_2012}, we take the $k$-LHP to be a maximization problem.} This framework of $k$-LHPs encapsulates many different physically motivated models, such as estimating the energy of quantum many-body systems, and thus has become an important problem to research. The $k$-LHP is QMA-hard to solve in general \cite{kempe_complexity_2006} and even in the case in which the local Hamiltonians interact on pairs of qudits (referred to as 2-local) \cite{cubitt_complexity_2014}. We focus our attention on solving these problems classically, which begs the question of how to efficiently represent a generic quantum solution to a $k$-LHP since quantum states can require $O(d^n)$ bits of information to write down in general. This problem is typically circumvented by only considering an efficiently representable subset of quantum states called an \emph{ansatz}. A very common ansatz is that of \emph{product-states} in which no entanglement is present and can be written using $O(poly(n,d))$ bits. Even with this vast simplification, algorithms outputting non-trivial solutions are still possible \cite{gharibian_approximation_2012,brandao_product-state_2016,parekh_beating_2021}. Unsurprisingly, allowing some entanglement allows one to improve on these bounds \cite{anshu_beyond_2020,parekh_application_2021,king_improved_2022,lee_optimizing_2022}. Nonetheless, it has been shown that product-state solutions are sufficiently good approximations of the optimal solution for dense graphs \cite{brandao_product-state_2016}.

The natural ``quantum" generalization of \textMCProb, known as \textQMCProb, can be seen as replacing assignments of colors to assignments of possibly entangled qubits. Then the notion of two assignments being different, the payoff function, is replaced by the energy of 2-local subsystems of qubits on the projector onto the antisymmetric subspace, which for two qubits is given by the projector onto the singlet state, \(\ket{\psi^-} = \frac{1}{\sqrt{2}}\ket{01} - \frac{1}{\sqrt{2}} \ket{10}\). Moreover, the natural ``quantum" generalization to \textMCdProb, known as \textQMCdProb, can be seen by replacing the \(d\) colors, \([d]\), with \(d\)-dimensional qudits and replacing the payoff with a projector onto the antisymmetric subspace of 2, \(d\)-dimensional qudits, which is the rank \(\binom{d}{2}\) projector onto the states \(\{\frac{1}{\sqrt{2}}\ket{ab} - \frac{1}{\sqrt{2}} \ket{ba}\ |\ a < b \in [d]\}\). We hope that our work on \textQMCdProb can help in the research of LHPs over qudits as \textMCProb and \textMCdProb did in the research of GCSPs. Moreover, with the increasing interest in high-dimensional quantum computing, i.e., quantum computers that use qudits \cite{luo_universal_2014,wang_qudits_2020}, optimization over qudits has become more of interest. 

\subsubsection*{Previous work}

Approximating ground states to \(k\)-LHPs over qubits with classical algorithms has a rich body of research, especially for the \textQMCProb problem. However, approximating ground states to \(k\)-LHPs over qudits is less rich. Gharibian and Kempe gave a polynomial-time approximation scheme (PTAS) for computing product-state solution to dense \(k\)-LHPs \cite{gharibian_approximation_2012}; their algorithm admits a \(d^{1-k}\)-approximation to the energy of the optimal product-state solution. After that, Brand\~ao and Harrow gave multiple PTASes for \(k\)-LHPs over qudits for planar graphs, dense graphs, and low threshold rank graphs \cite{brandao_product-state_2016}. 

Of particular interest to this paper is the work in finding product-state solutions to the \textQMCProb problem. In particular, the SDP-based algorithm by Bri\"et, de Olivera Filho, and Vallentin \cite{briet_positive_2010} gives a PTAS with an approximation ratio of 0.956 to the optimal product-state solution \cite{hwang_unique_2022}. Then, subsequent work by Gharibian and Parekh \cite{gharibian_almost_2019} give a similar PTAS based on an SDP derived from the noncommutative sum-of-squares (ncSoS) hierarchy, with an approximation ratio of 0.498 to the best (possibly entangled) state. Their rounding algorithm used the theory of the Bloch sphere, a unit sphere, \(S^2\), that corresponds bijectively to valid pure state density matrices. In this paper, we extend these results to the \textQMCdProb problem.

\subsubsection*{Our Results}

We present a noncommutative sum-of-squares SDP-based randomized approximation algorithm for the \textQMCdProb problem that finds a product-state solution of mixed states with bounded purity. We give the following theorem to describe its approximation ratio.

\begin{theorem}[Approximation Ratios For Approximating \textQMCdProb]\label{thm_1} There exists an efficient approximation algorithm for \textQMCdProb that admits an \(\alpha_d\)-approximation, where the constants \(\alpha_d\) (for \(d \geq 2\)) satisfy, 
\begin{enumerate}[label=(\roman*)]
    \item \label{thm_1_pt_1} \(\alpha_d > \frac{1}{2}\left(1-1/d\right)\)
    \item \label{thm_1_pt_2} \(\alpha_d - \frac{1}{2}\left(1-1/d\right) \sim \frac{1}{2 d^3}\)
    \item \label{thm_1_pt_3} \(\alpha_2 \geq 0.498767,\alpha_3 \geq 0.372995,\alpha_4 \geq 0.388478,\alpha_5 \geq 0.406128,\alpha_{10} \geq 0.450614,\alpha_{100} \geq 0.4950005\)
\end{enumerate}
\end{theorem}

\begin{remark}
    As a heuristic to judge our algorithm, we can observe that the energy achieved by the maximally mixed state over all vertices is no worse than \(\frac{1}{2}(1-1/d)\) times the optimal. \cref{thm_1}, \cref{thm_1_pt_1,thm_1_pt_2} show that our algorithm performs better than this trivial assignment for all \(d \geq 2\). Similar metrics are used classically, where a random assignment gives a \(1-1/d\) approximation in expectation for \textMCdProb. In this, we note the Frieze-Jerrum algorithm achieves a better additive bound of \(\Theta\left(\frac{\ln d}{d^2}\right)\) above \(1-1/d\) \cite{frieze_improved_1997}.
\end{remark}

We note that the when \(d=2\) case (i.e., \textQMCProb) our algorithm is identical to the Gharibian-Parekh algorithm. For all other cases, \(d \geq 3\), we show that our analysis is tight by providing an algorithmic gap that matches these ratios. While not stated here, we note that the exact values of \(\alpha_d\) for \(d \geq 3\) are given by an analytical formula.

\begin{theorem}[Algorithmic gap of \textQMCdProb]\label{thm_alg_gap} The approximation algorithm for \textQMCdProb that rounds to mixed product-states using the basic SDP has algorithmic gap \(\alpha_d\) for \(d \geq 3\).
\end{theorem}

This is rather interesting as for \textQMCProb, there is no known hard instance that gives an algorithmic gap of \(\alpha_2\) \cite{hwang_unique_2022}.

Additionally, we present an SDP-based algorithm in approximating the optimal product-state solution of \textQMCdProb. We give the following theorem to describe its approximation ratio.

\begin{theorem}[Approximation Ratios For Approximating The Optimal product-state Solution of \textQMCdProb]\label{thm_2} \textQMCdProb admits an \(\beta_d\)-approximation to the optimal product-state solution with respect to the basic SDP, where the constants \(\beta_d\) (for \(d \geq 2\)) satisfy, 
\begin{enumerate}[label=(\roman*)]
    \item\label{thm_2_pt_1} \(\beta_d = 2 \alpha_d\) for \(d \geq 3\)
    \item\label{thm_2_pt_2} \(\beta_2 \geq 0.956337\) \cite{briet_positive_2010,hwang_unique_2022}
\end{enumerate}
\end{theorem}

Furthermore, because we can relate the ratio of approximating the optimal product-state solution with that of approximating the maximal energy of \textQMCdProb by a factor of two, it follows that we also beat the metric of a random product-state solution for approximating the optimal product-state solution.

\begin{corollary}[Beating Random Assignment for Approximating The Optimal Product-State Solution of \textQMCdProb]\label{thm_2_b} \textQMCdProb admits an \(\beta_d\)-approximation to the optimal product-state solution with respect to the basic SDP, where the constants \(\beta_d\) (for \(d \geq 2\)) satisfy, 
    \begin{enumerate}[label=(\roman*)]
        \item \label{thm_2_b_pt_1} \(\beta_d > 1-1/d\)
        \item \label{thm_2_b_pt_2} \(\beta_d - \left(1-1/d\right) \sim \frac{1}{d^3}\)
\end{enumerate}
\end{corollary}


In achieving these results, we extend the methods of \cite{briet_positive_2010,gharibian_almost_2019}, which only work for \textQMCProb, to \textQMCdProb for arbitrary \(d \geq 2\) in two major ways. We first present an SDP based on the noncommutative sum-of-squares (ncSoS) Hierarchy for many-body systems over qudits using a generalization of the Pauli matrices called the generalized Gell-Mann matrices. While useful for defining observables, they lose many other useful properties of the Pauli matrices. Namely, they are not unitary (for \(d \geq 3\)), which makes defining the SDP difficult. In particular, arguing that the SDP vectors are unit vectors is not as trivial as observing that \(P^2 = I\), for a Pauli matrix \(P\), which requires \(P\) to be unitary. Our SDP shares many similarities to SDPs that have been used before \cite{barthel_solving_2012,brandao_product-state_2016,parekh_application_2021,parekh_beating_2021}, however, to work with our rounding algorithm we must make additional guarantees about our SDP, which allow us to simplify the SDP further. Second, we observe that the notion of a Bloch sphere does not exist for qudits of dimension \(d \geq 3\) \cite{bertlmann_bloch_2008} and so one can not round to the unit sphere in which the Bloch vectors for pure states lie. To get around this, we round to mixed states with bounded purity.

\subsubsection*{Paper Organization}

To show these results, we use a matrix basis that can be seen as a generalization of the Pauli matrices called the generalized Gell-Mann matrices, which we introduce in \cref{section_mat_basis} along with the Preliminaries in \cref{section_prelims}. This matrix basis will allow us to write the Hamiltonian in a convenient way, allowing for both an SDP relaxation and a Bloch vector representation of quantum states. With this, in \cref{section_qmc_ham}, we outline the main focus of our work, namely the \textQMCdProb problem and many special cases that aid in the analysis of our algorithm. Then, in \cref{prod_state_section}, we discuss the notion of Bloch vectors for qudits and discuss the critical geometric challenges of rounding to product-states that did not exist for qubits. In understanding these challenges, we discuss a resolution that will become the primary inspiration for our rounding algorithm. Following that, in \cref{section_sos}, we derive an SDP relaxation through the second level of the ncSoS hierarchy that also enforces all two body moments to be valid density matrices. Additionally, we make several new observations crucial for our rounding algorithm about the ncSoS hierarchy for higher dimensional qudits. Then, the main technical contribution of our work is in proving \cref{thm_1,thm_2}, which we do in \cref{section_rounding} by presenting our full algorithm and analyzing its approximation ratio for all \(d \geq 2\) using the Gaussian hypergeometric function. Then, we show our analysis is tight by providing an algorithmic gap instance in \cref{section_alg_gap} and proving \cref{thm_alg_gap}. Lastly, we give open problems and future directions in \cref{section_conclusion}.

\section{Preliminaries}\label{section_prelims}

We use the notation \([n] := \{1, \dotsc, n\}\). Let \(S_n\) denote the \emph{symmetric group} over the set \([n]\), which has order \(n!\) and \(A_n := \{\sigma \in S_n\ |\ \sgn(\sigma) = 1\}\) denote the \emph{alternating group}, which has order \(n!/2\). Let $\mathcal{H}_d := \mathbb{C}^d$ be a $d$-dimensional Hilbert Space, where a \emph{\(d\)-dimensional qudit} or a \emph{state} refers to a unit vector in $\mathcal{H}_d$. When \(d=2\), we call these states \emph{qubits}. In the case of qubits, it is common to represent the \textit{standard basis} as $\{\ket{0}, \ket{1}\}$, which represent the quantum analogs of classical bits. For qudits, we denote the \textit{standard basis} as \(\{\ket{i}\ |\ i \in [d]\}\). Let \(\mathcal{H}_d^{\otimes n}\) be the space of \(n\), \(d\)-dimensional qudits. We use \(\bra{\psi}\) to denote the \emph{conjugate transpose} of \(\ket{\psi} \in \mathcal{H}_d^{\otimes n}\) and \(\braket{\psi, \phi}\) to denote the inner product of two states. The \emph{standard basis} for \(\mathcal{H}_d^{\otimes n}\) is denoted by \(\{\ket{i_1 i_2 \dotsc i_n} := \ket{i_1} \otimes \ket{i_2} \otimes \cdots \otimes \ket{i_n}\ |\ i_1, i_2, \dotsc, i_n \in [d]\}\). 

We use the notation \(\mathcal{L}(\H_d^{\otimes n}) \cong \C^{d^n \times d^n}\) to denote the Hilbert-Schmidt space of linear operators from \(\H_d^{\otimes n}\) to itself, \(\H_d^{\otimes n} \ra \H_d^{\otimes n}\). For an operator/matrix $A \in \mathcal{L}(\H_d^{\otimes n})$, $A$ is \emph{positive semidefinite} or \emph{PSD} (denoted $A \succcurlyeq 0$) if $\forall \ket{\psi} \in \H_d^{\otimes n}, \bra{\psi} A \ket{\psi} \geq 0$. We use \(A^* := \overline{A^\top}\) to denote the \emph{adjoint/conjugate transpose}. A matrix, \(A\), is \emph{Hermitian} if \(A^* = A\). We use the notation \(\mathcal{D}(\H_d^{\otimes n}) := \{\rho \in \mathcal{L}(\H_d^{\otimes n})\ |\ \rho^* = \rho,\ \rho \succcurlyeq 0,\ \tr(\rho)=1\}\) to denote the subset of \emph{density matrices} on \(n\) qudits. A density matrix $\rho$ is \emph{pure} if it is a projector, namely, \(\rho^2 = \rho = \ket{\psi}\bra{\psi}\) (equivalently, \(\tr(\rho^2) = 1\)). More generally, for an arbitrary density matrix \(\rho \in \mathcal{D}(\H_d^{\otimes n})\), we refer to quantity \(\tr(\rho^2)\) as its \emph{purity}. Special cases of operators are the \emph{unitary group}, denoted \(U(d) := \{U \in \mathcal{L}(\H_d)\ |\ U^* U = I\}\). A subgroup of this group is the \emph{special unitary group}, denoted \(SU(d):= \{U \in U(d)\ |\ \det(U) = 1\}\). 

The \emph{symmetric subspace} of \(H_d^{\otimes n}\), the space of \(n\), \(d\)-dimensional qudits, is given by
\[\vee^n \H_d := \{\ket{\psi} \in H_d^{\otimes n}\ |\ \ket{\psi} = P(\sigma)\ket{\psi}\ \forall \sigma \in S_n\} = \text{span}\{\ket{\psi}^{\otimes n}\ |\ \ket{\psi} \in H_d\}\]
where \(P: S_n \ra GL(\H_d^{\otimes n})\) is the representation of \(S_n\) given by \(P(\sigma) := \sum_{i_1,\dotsc,i_n} \ket{i_{\sigma^{-1}(1)} \dotsc i_{\sigma^{-1}(n)}} \bra{i_1 \dotsc i_n}\). Similarly, the \emph{antisymmetric subspace} of \(H_d^{\otimes n}\), is given by
\[\wedge^n \H_d := \{\ket{\psi} \in H_d^{\otimes n}\ |\ \ket{\psi} = \sgn(\sigma) P(\sigma)\ket{\psi}\ \forall \sigma \in S_n\}\]
For \(n=2\), we note that these subspaces are orthogonal complements of each other. 

We generally use subscripts to indicate quantum subsystems. For instance, for a single qudit operator \(M \in \mathcal{L}(\H_d)\), we use \(M_i \in \mathcal{L}(\H_d^{\otimes n})\) to denoted the \(n\) qudit operator that acts as \(M\) on the \(i\)th qudit and identity on all other qudits else. Moreover, for a 2-qudit operator, \(H\), we use \(H_{ij}\) to denote that it acts on the \(i\) and \(j\) qudits (order matters). 

Let $G=(V,E)$ be an unweighted, undirected graph and \(G=(V,E,w)\) be a weighted, undirected graph. When viewing the vertices of a graph as a system of qudit, we use the notation $\mathcal{H}_d^{\otimes V}$. We denote by \(K_n\) a special graph called the \emph{complete graph on \(n\) vertices} (or the \emph{\(n\)-clique}) to be the graph with an edge between every pair of vertices. 
 
\begin{definition}[Algorithmic Gap]\label{def_alg_gap}
    Let \(\mathcal{P}\) be a maximization problem and \(A\) an approximation algorithm. For instance \(\mathcal{I}\) of \(\mathcal{P}\), let \(A(\mathcal{I})\) be the expected value of the solution outputted by the approximation algorithm and \(OPT(\mathcal{I})\) be the true optimal value. The algorithmic gap of the instance \(\mathcal{I}\) is the quantity
    \begin{equation}\label{alg_gap_inst}
        \text{Gap}_A(\mathcal{I}) = \frac{A(\mathcal{I})}{\text{OPT}(\mathcal{I})}
    \end{equation}
    The algorithmic gap of \(A\) for the problem \(\mathcal{P}\) is the quantity
    \begin{equation}\label{alg_gap_A}
    \inf_{\mathcal{I}} \left\{\text{Gap}_A(\mathcal{I})\right\}
\end{equation}
\end{definition}
The \emph{approximation ratio} is then a lower bound on this quantity. Namely, it's a constant, \(\alpha\), such that for all instances \(\mathcal{I}\), \(A(\mathcal{I}) \geq \alpha \text{OPT}(\mathcal{I})\).

Lastly, we use \(\sim\) to denote asymptotic equivalence. That is, two function, \(f,g: \R \ra \R\) are asymptotic equivalent, denoted \(f \sim g\), if \(\lim_{x\ra\infty} \frac{f(x)}{g(x)} = 1\) or equivalently if \(f(x) = g(x)(1 + o_x(1))\).

\subsection{Matrix Basis}\label{section_mat_basis}

When working in higher level qudit systems, it becomes less obvious how to describe quantum states and Hamiltonians. In the qubit case, we have the Pauli matrices, which have become immensely useful to the study of Hamiltonian optimization. For qudits, there are multiple options for matrix bases that generalize the Pauli matrices, maintaining different useful properties \cite{bertlmann_bloch_2008}. In this section, we look at one such basis called the \emph{generalized Gell-Mann matrix basis}. We also discuss another commonly used basis of unitary matrices basis called the clock and shift matrices in \cref{appendix_odd_deg_terms}, but we note that they ultimately fall short when using them as a basis for Bloch vectors.

\begin{definition}[Generalized Gell-Mann Matrices \cite{bertlmann_bloch_2008,kimura_bloch_2003,bossion_general_2021}]\label{gen_gell-mann_def} The \emph{generalized Gell-Mann matrices} (GGM matrices) are a higher dimensional extension of the Pauli Matrices and the Gell-Mann matrices \cite{gell-mann_symmetries_1962}. They form a real, orthogonal basis for \(d \times d\), hermitian, traceless matrices. They are given by the following \(d^2-1\) matrices, which can be divided into three categories.

\begin{itemize}
    \item Symmetric 
        \[\Lambda^+_{ab} := \ket{a}\bra{b} + \ket{b}\bra{a},\ \ \ \ 1 \leq a < b \leq d\]
    \item Antisymmetric 
        \[\Lambda^-_{ab} := -i\ket{a}\bra{b}  +i\ket{b}\bra{a},\ \ \ \ 1 \leq a < b \leq d\]
    \item Diagonal
        \[\Lambda^d_{a} :=  \sqrt{\frac{2}{a(a+1)}}   \left(\sum_{b=1}^{a} \ket{b}\bra{b}-a\ket{a + 1}\bra{a + 1}\right),\ \ \ \ 1 \leq a \leq d - 1\]
\end{itemize}

Often, it is more convenient to use a single index, which we will denote by \(\Lambda^a\). The relationship between these two definitions is given by the following indices, with \(1 \leq a < b \leq d\).

\begin{align}
    S_{ab}&=b^2+2(a-b)-1\\
    A_{ab}&=b^2+2(a-b)\\
    D_{a}&=a^2+2a
\end{align}
Then
\begin{equation}
    \Lambda^+_{ab} = \Lambda^{S_{ab}},\ \ \ \ \Lambda^-_{ab} = \Lambda^{A_{ab}},\ \ \ \ \Lambda^d_{a} = \Lambda^{D_{a}}
\end{equation}
Note that we will often switch between these two definitions depending on which model suits the situation best. However, most of the time, we will use the single index notation.
\end{definition}

\begin{remark}
    The Lie algebra \(\mathfrak{su}(d)\) is the real vector space of \(d \times d\), skew-hermitian, trace-less matrices with the Lie bracket. Moreover, \(\{i \Lambda^a\}_{a\in [d^2-1]}\) generates \(\mathfrak{su}(d)\). We use this fact to give the GGM matrices some algebraic structure. In fact, the choice of matrix basis is arbitrary, and it suffices to use any orthogonal basis that, under scalar multiple of \(i\), gives a generating set for \(\mathfrak{su}(d)\).
\end{remark}

We state some properties of the GGM matrices that we will use in this paper. First, they satisfy the following property: \(\tr(\Lambda^a \Lambda^b) = 2 \delta_{ab}\) (where \(\delta_{ab}\) is the Kronecker delta). Additionally, combined with the identity, or rather \(\Lambda^0 := \sqrt{\frac{2}{d}}I\), these \(d^2\) matrices give a real basis for all hermitian matrices, i.e., the space of single qudit observables, \(\mathcal{HM}_d\). It is for this reason they are well suited to be used to express Hamiltonian optimization problems. Lastly, they give a complex basis for \(\mathcal{L}(\C^d)\).

They have the following commutator/anti-commutator relations when viewed over \(\mathcal{L}(\C^d)\).

\begin{align}
    [\Lambda^a, \Lambda^b] &= \Lambda^a \Lambda^b - \Lambda^b \Lambda^a = 2 i \sum_{c = 1}^{d^2-1} f_{abc} \Lambda_c &a,b \in [d^2-1] \label{gell-mann_comm_rel}\\
    \{\Lambda^a, \Lambda^b\} &= \Lambda^a \Lambda^b + \Lambda^b \Lambda^a = \frac{4}{d}\delta_{ab}I + 2 \sum_{c = 1}^{d^2-1} d_{abc} \Lambda_c &a,b \in [d^2-1] \label{gell-mann_anticomm_rel}
\end{align}

Where the totally antisymmetric structure constants of \(\mathfrak{su}(d)\) are given by \(f_{abc}\) and the totally symmetric constants are given by \(d_{abc}\). We note that these constants have known closed formulas \cite{bossion_general_2021} that can be taken advantage of to define SDP constraints. Additionally, putting these together, we get the following product property. 
\begin{equation}\label{gell-mann_prod_property}
    \Lambda^a \Lambda^b = \frac{2}{d}\delta_{ab}I + \sum_{c = 1}^{d^2-1} d_{abc} \Lambda_c + i \sum_{c = 1}^{d^2-1} f_{abc} \Lambda_c
\end{equation}



\section{The \texorpdfstring{\textQMCdProb}{Quantum Max-d-Cut} Hamiltonian and the \texorpdfstring{\(SU(d)\)}{SU(d)} Heisenberg model}\label{section_qmc_ham}

In this paper, we focus on the \(d\)-dimensional qudit generalization of the well-studied Heisenberg model, called the \(SU(d)\) Heisenberg model with no local terms \cite{kawashima_ground_2007,beach_su_2009,lou_antiferromagnetic_2009}. This model is defined as having the edge interaction
\begin{equation}\label{general_local_SU(d)_heis_ham}
    h^{\text{Heis}_d} = \frac{1}{4}\sum_{a=1}^{d^2-1}\Gamma^a \otimes \Gamma^a
\end{equation}
where \(\{\Gamma^a\}_{a\in[d^2-1]}\) is some set of traceless, hermitian matrices satisfying \(\tr(\Gamma_a \Gamma_b) = 2 \delta_{ab}\), e.g., \(\{\Gamma^a\}_{a\in[d^2-1]} = \{\Lambda^a\}_{a\in[d^2-1]}\) are the generalized Gell-Mann matrices, defined in \cref{gen_gell-mann_def}. We note that the \(SU(d)\) Heisenberg model is known to be universal in the sense that any other Hamiltonian can be simulated by a Hamiltonian made from these \(SU(d)\) Heisenberg edge interactions \cite{piddock_universal_2021}.

Then, for a graph \(G = (V, E, w)\), we can define the full problem Hamiltonian to be \(H_G := \sum_{e \in E} w_e h_e\). Following the convention for Hamiltonian optimization problems, we define the ground state energy to the minimal eigenvalue (rather than the maximal energy as is the case for \textQMCProb).

We note that up to adding an identity term, \eqref{general_local_SU(d)_heis_ham} is nothing but the projector onto the symmetric subspace of \(\H_d^{\otimes 2} := (\C^d)^{\otimes 2}\). We can formalize this with the following proposition.

\begin{proposition}\label{prop_Psym_is_SUd_heis_plus_indent_term}
    For \(P_{\text{sym}}\) being the orthogonal projector onto the symmetric subspace of \(\H_d^{\otimes 2}\), we have that \(\frac{1}{2}\left(\frac{d+1}{d}\right)I + h^{\text{Heis}_d} = P_{\text{sym}}\).
\end{proposition}

With this, we can now define the \textQMCdProb edge interaction and Hamiltonian.

\begin{definition}[The \textQMCdProb edge interaction]\label{QMCd_edge_interaction_def} For 2-qudits, we define the \emph{\textQMCdProb edge interaction} to be the projector onto the antisymmetric subspace of two, \(d\)-dimensional qudits. 
Using an orthonormal basis for the antisymmetric subspace, we get that
\begin{equation}\label{QMCk_edge_interaction_eq}
    h := \sum_{1 \leq a< b \leq d} \left(\frac{1}{\sqrt{2}}\ket{ab} - \frac{1}{\sqrt{2}}\ket{ba}\right)\left(\frac{1}{\sqrt{2}}\bra{ab} - \frac{1}{\sqrt{2}}\bra{ba}\right)
\end{equation}
\end{definition}

Because this is a 2-qudit orthogonal projector onto the antisymmetric subspace, it is precisely the projector onto the complement of the symmetric subspace, i.e., \(h = I - P_{\text{sym}}\). 
Using \cref{prop_Psym_is_SUd_heis_plus_indent_term}, we get that the \textQMCdProb edge interaction can be decomposed in the following way.

\begin{proposition}[Alternate definition for the \textQMCdProb edge interaction]\label{QMCk_edge_interaction_gell-mann_prop}
The \textQMCdProb edge interaction can be written in terms of the generalized Gell-Mann matrices in the following way.
\begin{equation}\label{QMCk_edge_interaction_gell-mann_eq}
    h = \frac{1}{2}\left(\frac{d-1}{d}\right) I - \frac{1}{4}\sum_{a=1}^{d^2-1}\Lambda^a \otimes \Lambda^a
\end{equation}
\end{proposition}
We prove \cref{prop_Psym_is_SUd_heis_plus_indent_term,QMCk_edge_interaction_gell-mann_prop} in \cref{appendix_heis_model_and_QMCd}.

It follows from the definition that the \textQMCdProb edge interaction is naturally invariant under conjugation by local unitaries. Moreover, because \textQMCdProb is nothing but the negative of the \(SU(d)\) Heisenberg model and an identity term, finding its ground state, which in our case is the maximal eigenvalue, is \cc{QMA}-complete as was proven by Piddock and Montanaro \cite{piddock_universal_2021}.

\begin{definition}[\textQMCdProb]\label{QMCk_def}
    Let \(G = (V,E,w)\) be a graph with edge weights, which we refer to as the \emph{interaction graph}. The \emph{\textQMCdProb problem Hamiltonian}, \(H_G \in \mathcal{L}\left(\H_d^{\otimes V}\right)\), is given by
    \begin{equation}\label{QMCk_eq}
    H_G = \frac{1}{W}\sum_{(u,v) \in E} w_{uv} h_{uv} = \E_{(u,v) \sim E} h_{uv}
\end{equation}
    where \(W = \sum_{(u,v) \in E} w_{uv}\) denotes the sum of the weights and \(h_{uv} \in \mathcal{L}\left(\H_d^{\otimes V}\right)\) denotes the \textQMCdProb edge interaction applied to the qudits \(u\) and \(v\), namely, \(h_{uv} \otimes I_{V \setminus \{u,v\}}\).
\end{definition}

Equivalently, we can think of the interaction graph as giving a distribution over the local Hamiltonians. Then, we ask to optimize the expected energy over the local Hamiltonians. With that, we now give definitions to express the main focus of this work.

\begin{definition}[Energy of \textQMCdProb]\label{QMCk_energy_def} Let \(H_G\) be an instance of \textQMCdProb. The \emph{energy} of a state \(\ket{\psi} \in \H_d^{\otimes V}\), is the quantity \(\bra{\psi} H_G \ket{\psi}\). The maximum energy, also referred to as the \emph{value}, of \(H_G\) is
    \[\QMCdProb(G) = \lambda_\text{max}(H_G) = \max_{\substack{\ket{\psi} \in \H_d^{\otimes V}\\\braket{\psi|\psi} = 1}} \bra{\psi} H_G \ket{\psi}\]
    Equivalently, we can define this using density matrices.
    \[\QMCdProb(G) = \max_{\substack{\rho \in \mathcal{D}\left(\H_d^{\otimes V}\right)\\\text{s.t. }\rho^2 = \rho}} \tr(\rho H_G)\]
\end{definition}

We note that the constraint \(\rho^2 = \rho\) requires that the density matrices represent pure states. We make this constraint explicit because we will also look at a relaxation of \textQMCdProb by allowing for ancilla bits, which will, in turn, allow us to drop this constraint.

\begin{definition}[Energy of \textQMCdProb with ancillas]\label{QMCk_with_ancillas_energy_def} Let \(H_G\) be an instance of \textQMCdProb. Given a state with \(r \geq 0\) ancilla bits, \(\ket{\psi} \in \H_d^{\otimes V} \otimes \H_d^{\otimes r}\), it's energy is the quantity \(\bra{\psi} H_G \otimes I_A \ket{\psi}\), for which, we denote with subscript \(A\) the space of \(r\) additional ancilla bits. The maximum energy of \(H_G\) is also given by
    \[\begin{split}
        \QMCdProb_{\text{ancilla}}(G) = \max\ &\bra{\psi} H_G \otimes I_A \ket{\psi}\\
        {\scriptstyle\text{s.t. }} & {\scriptstyle\ket{\psi} \in \H_d^{\otimes V} \otimes \H_d^{\otimes r} \text{ for } r \geq 0} \\
        & {\scriptstyle\braket{\psi|\psi} = 1}
    \end{split}\]
\end{definition}

We can then make the following observations about this new problem.

\begin{proposition} The problem of optimizing over states with ancilla qubits is equivalent to optimizing over mixed density matrices.
    \[\QMCdProb_{\text{ancilla}}(G) = \max_{\rho \in \mathcal{D}\left(\H_d^{\otimes V}\right)} \tr(\rho H_G)\]
\end{proposition}


\begin{remark}
    While allowing for ancillas can be seen as a relaxation of \textQMCdProb, the maximal energy is the same in both problems. So, in a sense, it is only a relaxation of the search space. Therefore, we will often refer to both definitions interchangeably with \(\QMCdProb(G)\).
\end{remark}

For the purposes of rounding, we will consider two special cases of \textQMCdProb that restrict the search space to be over product states. This is a common first step in studying algorithms for Local Hamiltonian problems as it is a well-understood ansatz. The first of which is when the solution space is over product states of pure states. This is the ansatz originally used to study the \textQMCProb problem \cite{gharibian_almost_2019,hwang_unique_2022,parekh_optimal_2022} as well as general LHPs \cite{gharibian_approximation_2012,brandao_product-state_2016} and has been used as a starting point for many subsequent algorithms \cite{parekh_application_2021,king_improved_2022}.

\begin{definition}[Pure product state value]\label{def_pure_prod_prob} The product state value of pure states of \(H_G\) is
    \[\textsc{PureProd}_{\textsc{QMC}_d}(G) = \max_{\substack{\forall v \in V, \rho_v \in \mathcal{D}(\C^d)\\\text{s.t. }\rho_v^2 = \rho}} \tr(\rho_G H_G)\]
    where \(\rho_G = \bigotimes_{v\in V} \rho_v\).
\end{definition}

Next, we will consider a product state solution of mixed states. Unlike before, in \cref{QMCk_with_ancillas_energy_def}, we won't look at a mere relaxation of the pure state variant, and instead, we will restrict to a certain level of mixed states that would no longer include pure states.

\begin{definition}[Mixed product state value] \label{MixedProd_def} The product state value of mixed states of \(H_G\) is
    \[\textsc{MixedProd}_{\textsc{QMC}_d}(G) = \max_{\substack{\forall v \in V, \rho_v \in \mathcal{D}(\C^d)\\\text{s.t. }\tr(\rho_v^2) = \frac{1}{d-1}}} \tr(\rho_G H_G)\]
    where \(\rho_G = \bigotimes_{v\in V} \rho_v\).
\end{definition}

We note that the constraint \(\tr(\rho_v^2) = \frac{1}{d-1}\) is used to specify the ``amount" of mixture that we want our states to have. This quantity is often referred to as the purity of a density matrix \(\rho\). The significance of this will be made clear in \cref{prod_state_section,prop_tr_rho2_vs_bloch_vec}. For now, we note that in the \(d=2\) case, this is nothing but \(\textsc{PureProd}_{\textsc{QMC}_2}(G)\). Secondly, we note that the maximally mixed state has a purity of \(\tr\left(\frac{1}{d^2} I\right) = \frac{1}{d}\) and so for large \(d\), we are optimizing over states that are close to maximally mixed, especially considering that pure states give \(\tr(\ket{\psi}\bra{\psi}^2) = 1\). Nonetheless, as stated in \cref{thm_1}, \cref{thm_1_pt_1}, we show that our algorithm achieves an improvement over a random assignment for all \(d \geq 2\), which in expectation is the same as the energy of the maximally mixed state. However, this improvement rapidly diminishes as stated in \cref{thm_1}, \cref{thm_1_pt_2}.

\section{Bloch Vectors and Product States}\label{prod_state_section}

The use of the Bloch vectors is very important in the rounding algorithm by Gharibian and Parekh \cite{gharibian_almost_2019}. In particular, single qubit pure states correspond bijectively to the Bloch vectors on a sphere. More generally, the set of valid Bloch vectors gives a unit ball, called the Bloch ball, whose vectors correspond bijectively to valid density matrices, with the pure states being on the surface and mixed states being inside the ball. This picture falls apart for qudits. In this section, we look at the Bloch vectors more closely and their extension to qudits for use in a generalization of the Gharibian-Parekh rounding algorithm.

Bloch vectors for qudits have been extensively studied \cite{jakobczyk_geometry_2001,kimura_bloch_2003,bertlmann_bloch_2008,goyal_geometry_2016}. We present an overview of the geometry of the convex set of Bloch vectors, denoted \(\Omega_d\), and discuss its implications for rounding. Then, we will look at product state solutions for \textQMCdProb. While all very standard linear algebra, we give proofs for all propositions in \cref{appendix_bloch_vecs} to be more self-contained. For a more complete discussion (albeit without proofs), see \cite{bengtsson_space_2006}.

Given some density matrix, \(\rho \in \mathcal{D}(\C^d)\), we can uniquely decompose it in the following way.
\begin{equation}\label{basic_bloch_vec_decomp}
    \rho = \frac{1}{d}I + \sum_{a = 1}^{d^2 - 1} b_a \cdot \Gamma^a
\end{equation}
Where the set \(\{\Gamma^a\}_a\) is a basis for traceless, hermitian matrices. We call \(\vec{b} \in \Omega_d \subseteq \R^{d^2-1}\) the Bloch vector. The matrix \(\rho^* = \frac{1}{d}I\) is a special density matrix called \emph{the maximally mixed state}. We note the following facts about Bloch vectors for qudits when \(\{\Gamma^a\}_a\) satisfies \(\tr(\Gamma^a \Gamma^b) = 2 \delta_{ab}\) (which is true for \(\{\Gamma^a\}_a = \{\Lambda^a\}_a\)). 

\begin{proposition}[Outsphere/Circumsphere]\label{outsphere_prop} Given a density matrix \(\rho\) and it's corresponding Bloch vector \(\vec{b}\), we always have that \(\|\vec{b}\| \leq \sqrt{\frac{d-1}{2d}}\). Moreover, if \(\rho\) is a pure state, i.e., \(\rho^2 = \rho\), then \(\|\vec{b}\| = \sqrt{\frac{d-1}{2d}}\).
\end{proposition}

We refer to the sphere of radius \(\sqrt{\frac{d-1}{2d}}\) as the outsphere or circumsphere, which defines the minimal ball that contains \(\Omega_d\). All pure states lie on the surface of this ball, while all mixed states are inside the ball.

\begin{remark}
    While every valid density matrix has a unique Bloch vector, it is not the case that every vector, \(\vec{b}\) such that \(\|\vec{b}\| \leq \sqrt{\frac{d-1}{2d}}\), gives a valid density matrix, i.e., \(\Omega_d\) is not a ball for \(d \geq 3\).\footnote{This is because the resulting matrix might not be PSD. This is easy to see in the \(k=3\) case with \(\frac{1}{3} I + \frac{1}{\sqrt{3}}\Lambda^d_{1}\), which has an eigenvalue of \(\frac{1}{3}-\frac{1}{\sqrt{3}} < 0\).\vspace{-0.05in}}
\end{remark}

The region of valid pure states is relatively small. This is characterized by the fact that the set of Bloch vectors that correspond to pure states is diffeomorphic to \(\C P^{d-1}\), the complex projective space in \(d\) complex dimensions, which, as a real manifold, has dimension \(2(d-1)\). This is in contrast to the outsphere, which is a sphere in \(\R^{d^2-1}\). So, rounding to pure states requires a more complete understanding of the geometry of pure states. 
However, within the convex region of valid Bloch vectors, \(\Omega_d\), there is a well-known maximal ball, which gives a sphere that we can still round to.

\begin{proposition}\label{insphere_prop} Within \(\Omega_d\), there is a maximal ball of radius \(\frac{1}{\sqrt{2d(d-1)}}\) that consists entirely of Bloch vectors that correspond to valid density matrices.
\end{proposition}

We call the surface of this ball the insphere. To relate this back to \cref{MixedProd_def}, we give the following.

\begin{proposition}\label{prop_tr_rho2_vs_bloch_vec}
    For a density matrix, \(\rho\), with corresponding Bloch vector, \(\vec{b}\), we have that \(\tr(\rho^2) = \frac{1}{d-1}\) if and only if \(\|\vec{b}\| = \frac{1}{\sqrt{2d(d-1)}}\).
\end{proposition}

We can observe that for \(d=2\), the outsphere is equal to the insphere, and so the Gharibian-Parekh rounding algorithm can equally be seen as rounding to Bloch vectors on the insphere. For this reason, our rounding algorithm is a natural generalization of the Gharibian-Parekh rounding algorithm. 

\subsection{Analysis of Product state solutions}\label{analysis_prod_state_section}

In this section, we look at the energy of \textQMCdProb for pure states and mixed product state solutions. We start with an analysis of product state solutions made up of pure states.

As described above, all pure states have Bloch vectors with \(\|\vec{b}\| = \sqrt{\frac{d-1}{2d}}\). Because it is useful to work with unit vectors, we rescale \(\Omega_d\) by a factor of \(\sqrt{\frac{2d}{d-1}}\), which gives the new equation for the Bloch vector decomposition as
\begin{equation}\label{rescaled1_bloch_vec_decomp}
    \rho = \frac{1}{d}I + \sqrt{\frac{d-1}{2d}} \sum_{a=1}^{d^2-1} b_a \Lambda^a
\end{equation}
Now, pure density matrices have unit Bloch vectors. We next give the following important proposition.
\begin{proposition}[\cite{jakobczyk_geometry_2001,byrd_characterization_2003}]\label{bloch_vec_angle_prop}
    For two pure states, \(\rho\) and \(\rho'\), with Bloch vectors \(\vec{b}\) and \(\vec{b}'\) as defined in \eqref{rescaled1_bloch_vec_decomp}, we have that \(\braket{\vec{b},\vec{b}'} \geq -\frac{1}{d-1}\). Furthermore, these pure states are orthogonal, i.e., \(\tr(\rho \rho') = 0\), if and only if \(\braket{\vec{b},\vec{b}'} = -\frac{1}{d-1}\).
\end{proposition}

For \(d=2\), this gives us the property that antipodal points correspond to orthogonal states. And in general, for \(d \geq 2\), this lines up with the classical \textMCdProb problem \cite{frieze_improved_1997} where the colors are represented by vertices on a simplex, which also have this same value for the inner product between different vertices. In the quantum setting, this notion is referred to as the eigenvalue simplex. 

Next, we look at the energy achieved by a product state solution of pure states. We consider the case of having a pure product state solution.

\begin{proposition}\label{prop_pureprod_leq_1/2}
    For any graph \(G\), we always have that \(\textsc{PureProd}_{\textsc{QMC}_d}(G) \leq 1/2\).
\end{proposition}

We note that in contrast, \(\QMCdProb(G)\) can be as large as 1.

\begin{proof}
For each \(v\in V\) let \(\vec{b}_v \in S^{d^2 - 2}\) be the Bloch vector for \(\rho_v\) as given by the decomposition in \eqref{rescaled1_bloch_vec_decomp}. Then the energy of \(\rho_G := \bigotimes_{v \in V} \rho_v\) is given by
\begin{equation}
\tr(\rho_G H_G) = \E_{(u,v) \sim E} \tr\left(h_{uv} \rho_u \otimes \rho_v\right)
\end{equation}

We then consider only a single one of these edge interaction terms to get that
\begin{align*}
    \tr\left(h_{uv} \rho_u \otimes \rho_v\right) &= \tr\left(\left(\frac{d-1}{2d} I - \frac{1}{4} \sum_{a=1}^{d^2-1} \Lambda^a \otimes \Lambda^a\right) \left(\frac{1}{d} I + \sqrt{\frac{d-1}{2d}}\; \vec{b}_u \cdot \vec{\Lambda}\right) \otimes \left(\frac{1}{d} I + \sqrt{\frac{d-1}{2d}}\; \vec{b}_v \cdot \vec{\Lambda}\right)\right) \\
    &= \frac{d-1}{2d}\tr\left(\frac{1}{d^2}I\right) - \frac{d-1}{8d} \sum_{a=1}^{d^2-1} \vec{b}_u(a) \vec{b}_v(a) \tr(\Lambda^2_a \otimes \Lambda^2_a) \\
    &= \frac{1}{2}\left(\frac{d-1}{d}\right)\left(1-\braket{\vec{b}_u, \vec{b}_v}\right)
\end{align*}

We can then upper bound this using \cref{bloch_vec_angle_prop} which gives us that \(\tr\left(h_{uv} \rho_u \otimes \rho_v\right) \leq \frac{1}{2}\). And putting that all together, we get that the energy of \(\ket{\psi_G}\) is equivalently given by

\[\tr(\rho_G H_G) = \frac{1}{2}\left(\frac{d-1}{d}\right)\E_{(u,v) \sim E} \left(1-\braket{\vec{b}_u, \vec{b}_v}\right) \leq \frac{1}{2}\]
\end{proof}

Note, unlike for the qubit case, this does not easily allow us to re-express the optimization over Bloch vectors because the set of pure states makes up a small portion of the outsphere. We can, of course, restrict to the subset of \(S^{d^2 - 2}\), which are Bloch vectors that correspond to valid density matrices, as we do in the following proposition, but that does not lend well to a rounding algorithm.

\begin{proposition}\label{prop_pure_prod_prob_rewrite} We can rewrite the pure product state value as follows, where 
\begin{equation*}
    \Omega_d^{\textsc{ext}} = \{\vec{b} \in S^{d^2-2}\ |\ \vec{b} \star \vec{b} = \vec{b}\} \subset \Omega_d
\end{equation*}
are the extremal points, i.e., those that correspond to pure states. We use \((\vec{b} \star \vec{b})_{c} := \sqrt{\frac{2d}{d-1}}\sum_{a,b = 1}^{d^2-1} d_{abc} b_a b_b\) to denote the quantity referred to as the star product (where \(d_{abc}\) are the symmetric constants of the matrix basis using in the Bloch vector decomposition \eqref{rescaled1_bloch_vec_decomp}).
    \begin{equation}
        \textsc{PureProd}_{\textsc{QMC}_d}(G) = \max_{f:V\ra \Omega_d^{\textsc{ext}}} \frac{1}{2}\left(\frac{d-1}{d}\right)\E_{(u,v) \sim E} \left(1-\braket{f(u), f(v)}\right)
    \end{equation}
\end{proposition}
\begin{proof}[Proof (sketch)]
    This follows from the steps used in the proof of \cref{prop_pureprod_leq_1/2} combined with the fact that Bloch vectors for pure states are exactly the set \(\Omega_d^{\textsc{ext}} = \{\vec{b} \in S^{d^2-2}\ |\ \vec{b} \star \vec{b} = \vec{b}\}\) \cite{goyal_geometry_2016,bengtsson_space_2006,byrd_characterization_2003}.
\end{proof}

Regardless of whether we could round to a pure product state solution, we can't hope to get a rounding algorithm with an approximation ratio of more than \(1/2\). The situation is worse when we consider the product state solution of mixed states. Specifically, we consider states whose Bloch vectors lie on the maximal sphere in \(\Omega_d\), which has radius \(\sqrt{\frac{1}{2d(d-1)}}\). Again, it is useful to work with unit vectors, so we rescale \(\Omega_d\) by a factor of \(\sqrt{2d(d-1)}\), which gives the new equation for the Bloch vector decomposition as
\begin{equation}\label{rescaled2_bloch_vec_decomp}
    \rho = \frac{1}{d}I + \frac{1}{\sqrt{2d(d-1)}} \vec{b} \cdot \vec{\Lambda}
\end{equation}

Now, any density matrix with a purity of \(\tr(\rho^2) = \frac{1}{d-1}\) will have a unit Bloch vector. We note that pure states in this model have Bloch vectors with \(\|\vec{b}\| = d-1\). Additionally, we note that we no longer have the property that \(\braket{\vec{b}_u, \vec{b}_v} \geq -\frac{1}{d-1}\) when \(\vec{b}_u, \vec{b}_v \in S^{d^2-2}\). Importantly, this means we lose the notion of orthogonality among these mixed states.

Now, we can look at the energy achieved by a product state solution of mixed states with bounded purity. Namely, \(\rho_G = \bigotimes_{v \in V} \rho_v\), where \(\rho_v\) is a mixed state which lies on the maximal sphere of density matrices, i.e., \(\tr(\rho_v^2) = \frac{1}{d-1}\) for all \(v \in V\).

\begin{proposition}\label{prop_rewite_mixed_state_and_upper_bound}
    We can rewrite the mixed product state value as follows.
    \begin{equation}
        \textsc{MixedProd}_{\textsc{QMC}_d}(G) = \max_{f:V \ra S^{d^2-2}} \frac{1}{2}\left(\frac{d-1}{d}\right)\E_{(u,v) \sim E} \left(1-\frac{\braket{f(u), f(v)}}{(d-1)^2}\right)
    \end{equation}
    Furthermore, we have that \(\textsc{MixedProd}_{\textsc{QMC}_d}(G) \leq \frac{1}{2} - o_d(1)\).
\end{proposition}

\begin{proof}
For each \(v\in V\) let \(\vec{b}_v \in S^{d^2 - 2}\) be its Bloch vector as given by the decomposition in \eqref{rescaled2_bloch_vec_decomp}. Then, as we did for \cref{bloch_vec_angle_prop}, we consider the energy for only a single edge interaction to get that

\begin{align*}
    \tr\left(h_{uv} \rho_u \otimes \rho_v\right) &= \tr\left(\left(\frac{d-1}{2d} I - \frac{1}{4} \sum_{a = 1}^{d^2-1} \Lambda^a \otimes \Lambda^a\right) \left(\frac{1}{d} I + \frac{1}{\sqrt{2d(d-1)}}\; \vec{b}_u \cdot \vec{\Lambda}\right) \otimes \left(\frac{1}{d} I + \frac{1}{\sqrt{2d(d-1)}}\; \vec{b}_v \cdot \vec{\Lambda}\right)\right) \\
    &= \frac{d-1}{2d}\tr\left(\frac{1}{d^2}I\right) - \frac{1}{8d(d-1)} \sum_{a=1}^{d^2-1} \vec{b}_u(a) \vec{b}_v(a) \tr((\Lambda^a)^2 \otimes (\Lambda^a)^2) \\
    &= \frac{1}{2}\left(\frac{d-1}{d}\right)\left(1-\frac{\braket{\vec{b}_u, \vec{b}_v}}{(d-1)^2}\right)
\end{align*}

This expression is maximized when \(\braket{\vec{b}_u, \vec{b}_v} = -1\), which gives us that 
\begin{equation}\label{opt_ratio_for_mix_prod_state}
    \tr\left(h_{uv} \rho_u \otimes \rho_v\right) \leq \frac{(d-1)^2+1}{2d(d-1)} = \frac{1}{2} - \frac{d-2}{2d(d-1)} = \frac{1}{2} - o_d(1)
\end{equation}

Then the energy of the mixed product state solution \(\rho_G = \bigotimes_{v \in V} \rho_v\) is given by.

\[\tr(H_G \rho_G) = \frac{1}{2}\left(\frac{d-1}{d}\right)\E_{(u,v) \sim E} \left(1-\frac{\braket{\vec{b}_u, \vec{b}_v}}{(d-1)^2}\right) \leq \frac{1}{2} - o_d(1)\]

Finally, we can optimize over the Bloch vectors \(\{\vec{b}_u\}_{v \in V}\), because they are in bijection with density matrices such that \(\tr(\rho^2) = \frac{1}{d-1}\), as shown by \cref{insphere_prop,prop_tr_rho2_vs_bloch_vec}, and the fact that the GGM matrices and the identity form a real basis for Hermitian matrices.
\end{proof}

\begin{remark}
    While the approximation ratio for any algorithm that gives these mixed states is bounded by \(\frac{1}{2} - o_d(1)\), we note that this is not too far from the \(\frac{1}{2}\) bound for pure states. At it's worse, when \(d=3,4\) we have that the energy is at most \(\frac{5}{12} \approx 0.41667\).
\end{remark}

Lastly, we look at the energy of a random assignment of pure states, or rather the maximally mixed state \(\rho_G^* = \frac{1}{d^n} I = \bigotimes_{u \in V} \rho^*\), where \(\rho^* = \frac{1}{d} I\) we use to denote the maximally mixed state of a single qudit. This can be seen as analogous to the classical random assignment. Furthermore, it is also a product-state solution. 

\begin{proposition}
     The energy of the maximally mixed states is \(\frac{1}{2}\left(1-\frac{1}{d}\right)\).
\end{proposition}
\begin{proof}
    We consider the maximally mixed state over all vertices, which is given by \(\rho^*_G = \bigotimes_{u \in V} \rho^*_v\).
    We first consider the energy of a single-edge interaction.
    \begin{align*}
        \tr\left(h_{uv} \rho_* \otimes \rho_*\right) &= \tr\left(\left(\frac{d-1}{2d} I - \frac{1}{4} \sum_{a = 1}^{d^2-1} \Lambda^a \otimes \Lambda^a\right) \left(\frac{1}{d} I\right) \otimes \left(\frac{1}{d} I \right)\right) \\
        &=\frac{d-1}{2d}\tr\left(\frac{1}{d^2}I\right) - \frac{1}{4d^2} \sum_{a = 1}^{d^2-1} \tr(\Lambda^a \otimes \Lambda^a) \\
        &= \frac{d-1}{2d} = \frac{1}{2}\left(1-\frac{1}{d}\right)
    \end{align*}
which can then be extended to the full graph to get
\(\tr\left(H_G \rho_*^{\otimes V}\right) = \frac{1}{2}\left(1-\frac{1}{d}\right)\).
\end{proof}

Combined with the trivial upper bound of \(\QMCdProb(G) \leq 1\), we note that this gives an approximation no worse than \(\frac{1}{2}\left(1-\frac{1}{d}\right)\).

\section{Sum-of-Squares Hierarchies and SDP relaxations}\label{section_sos}

In this section, we derive the Semidefinite Program (SDP) that we used to relax the \textQMCdProb problem through the framework of noncommuting Sum of Squares (ncSoS) commonly referred to as the quantum Lasserre hierarchy. We note that much of this framework generalizes existing derivations of SDPs that optimize Hamiltonians over qubits \cite{barthel_solving_2012,baumgratz_lower_2012,brandao_product-state_2016,gharibian_almost_2019,parekh_application_2021,parekh_beating_2021,king_improved_2022,lee_optimizing_2022,hwang_unique_2022}.

\begin{definition}[General 2-Local Hamiltonian Problem]\label{basic_2LHP_def}
    The basic \(2\)-local Hamiltonian optimization problem asks to maximize the expected energy over a distribution, given by an interaction graph, \(G = (V,E,w)\), of local Hamiltonians, \(\{h_{uv}\ |\ (u,v) \in E\}\). That is, we wish to find the quantity of the following expression.
    \begin{equation}\label{basic_LHP}
        \text{OPT(G)} = \max_{\rho \in \mathcal{D}\left(\H_d^{\otimes n}\right)} \E_{(u,v) \sim E} \left[\tr\left(\rho h_\alpha\right)\right]
    \end{equation}
    Additionally, in this section, we make two assumptions about the local Hamiltonians, \(\{h_{uv}\ |\ (u,v) \in E\}\), namely that they are PSD and that they are invariant under conjugation of local unitaries.  
\end{definition}

Using a basis, \(\{\Lambda^a\}_{a \in [d^2-1]}\), of traceless, hermitian operators for a single qudit such that, and the identity such that \(\tr(\Lambda^a \Lambda^b) = 2 \delta_{ab}\) (e.g., the generalized Gell-Mann matrices), we define the following basis for hermitian operators over \(n\), \(d\)-dimensional qudits.
\begin{equation}\label{P_d^n}
    \mathcal{P}_d^n := \{\Lambda^a\ |\ a \in \{0, 1,2, \dotsc, d^2-1\}\}^{\otimes n}
\end{equation}
We again note that \(\Lambda^0 := \sqrt{\frac{2}{d}}I\) is used to denote the normalized identity.

\begin{remark}\label{remark_reababilty_but_more_confusing_Lambda0_I}
    Due to how we defined the subscript notation, we note that \(\Lambda^a_u := \Lambda^a_u \otimes I_{[n] \setminus \{u\}} \notin \mathcal{P}_d^n\) is not in \eqref{P_d^n}. It is instead a scalar multiple of the basis element \(\Lambda^a_u \otimes \Lambda^0_{[n] \setminus \{u\}}\). Nonetheless, for the sake of readability, we often will conflate the two, using the fact that \(I = \sqrt{\frac{d}{2}} \Lambda^0\), to simplify many of our expressions.
\end{remark}

Using \eqref{P_d^n}, a local Hamiltonian optimization problem can be expressed as an optimization over real vectors \(\vec{y} \in \R^{d^{2n}}\) such that \(\rho_{\vec{y}} := \frac{1}{2^n}\sum_{A \in \mathcal{P}_d^n} y(A) A\) is a valid density matrix (in particular, that \(\rho_{\vec{y}} \succcurlyeq 0\) and that \(\tr(\rho_{\vec{y}}) = 1\)). Conversely, given a density matrix \(\rho \in \mathcal{D}\left(\H_d^{\otimes n}\right)\), we can define the vector such that \(y_\rho(A) := \tr(\rho A)\) for \(A \in \mathcal{P}_d^n\). To influence how we define the SDP, we consider the bilinear form, \mbox{\(M: \mathcal{L}(\H_d^{\otimes n}) \times \mathcal{L}(\H_d^{\otimes n}) \ra \C\)}, defined below, that will become our moment matrix.
\begin{equation} \label{SDP_mat_goal}
M(A, B) := \tr(\rho A^* B)
\end{equation}
As a matrix, this is nothing but \(M = \left(\tr(\rho A^* B)\right)_{A,B \in \mathcal{P}_d^n}\).
We note that this gives an equivalent definition of a density matrix (namely, \(M\) such that \(M \succcurlyeq 0\), it knows about matrix multiplication, i.e., \mbox{\(M(A,B) = M(A',B')\)} such that \(A^*B = A'^*B'\), and that \(M(I,I) = 1\)). Then, this gives us an equivalent way to optimize; namely, we optimize \(M(I,H)\) over bilinear forms, \(M\), that satisfy these properties.

We note that because our basis is of exponential size, we are optimizing over an exponential number of variables. To counteract this, we consider the notion of a pseudo-density matrix. First, we define a new basis that bounds the degree of the basis elements. 

\begin{definition}
    For a basis element/monomial \(A \in \mathcal{P}_d^n\), we define its \emph{degree}, denoted \(\omega(A)\), to be the number of non-identity terms (up to scalar multiples) or the number of qudits such that \(A\) doesn't act as the identity (up to scalar multiples). For example \(\Lambda^a_u \otimes \Lambda^b_v \otimes \Lambda^0_{[n] \setminus \{u,v\}}\) has degree 2 assuming \(a,b \neq 0\). We extend this notion to an arbitrary operator to be the largest degree among all of its non-zero components.
\end{definition}

Here, we can think of an arbitrary operator \(B \in \H_d^{\otimes n}\) as a polynomial. We then define the following basis for the subspace of operators with degree/degree at most \(t\), denoted \(\mathcal{L}^{(t)}(\H_d^{\otimes n})\).
\begin{equation}\label{P_d^n(d)}
    \mathcal{P}_d^n(t) := \{A \in \mathcal{P}_d^n\ |\ \omega(A) \leq t\}
\end{equation}

\begin{definition}[degree-\(2t\) pseudo-density Matrix] 
    A matrix \(\Tilde{\rho} \in \mathcal{L}(\H_d^{\otimes n})\) is called a \emph{degree-\(2t\) pseudo-density matrix} over \(n\), \(d\)-dimensional qudits if \(\tr(\Tilde{\rho}) = 1\), \(\Tilde{\rho}^* = \Tilde{\rho}\), and \(\tr(\Tilde{\rho} A^* A) \geq 0\) for all \(A \in \mathcal{L}^{(t)}(\H_d^{\otimes n})\). Additionally, we denote the space of degree-\(2t\) pseudo-density matrices as 
    \[\Tilde{\mathcal{D}}^{(2t)}(\H_d^{\otimes n}) := \{\Tilde{\rho} \in \mathcal{L}(\H_d^{\otimes n})\ |\ \tr(\Tilde{\rho}) = 1,\ \Tilde{\rho}^* = \Tilde{\rho},\ \tr(\Tilde{\rho} A^* A) \geq 0 \text{ for all } A \in \mathcal{L}^{(t)}(\H_d^{\otimes n})\}\]
\end{definition}

This can be most easily seen as a relaxation of the PSD constraint on the density matrices, in particular: 

\begin{proposition}[degree-\(2n\) pseudo-density Matrices are Valid Density Matrices] 
    \(\Tilde{\mathcal{D}}^{(2n)}(\H_d^{\otimes n}) = \mathcal{D}(\H_d^{\otimes n})\).
\end{proposition}

And so, we have the following hierarchy that indicates how ncSoS relaxes the problem of optimizing over valid density matrices.
\begin{equation}\label{pdensity_hierarchy}
    \Tilde{\mathcal{D}}^{(2)}(\H_d^{\otimes n}) \supset \Tilde{\mathcal{D}}^{(4)}(\H_d^{\otimes n}) \supset \cdots \supset\Tilde{\mathcal{D}}^{(2(n-1))}(\H_d^{\otimes n}) \supset \Tilde{\mathcal{D}}^{(2n)}(\H_d^{\otimes n}) = \mathcal{D}(\H_d^{\otimes n})
\end{equation}

\begin{remark}
    While a pseudo-density matrix is defined to be Hermitian, and thus it can be written as a real combination of \(\mathcal{P}_d^n\), the definition says nothing about the components of terms in \(\mathcal{P}_d^n \setminus \mathcal{P}_d^n(2t)\). Furthermore, because we are working with 2 local Hamiltonians and thus the objective does not depend on them either (even for \(t = 1\)), the values of these components can be ignored or, without loss of generality, can be set to 0. In other words, we can assume that the pseudo-density matrix \(\Tilde{\rho} \in \mathcal{L}^{(2t)}(\H_d^{\otimes n})\) has bounded degree. 
\end{remark}

All in all, this allows us to rephrase our optimization over vectors \(\vec{y} \in \R^{O(n^t d^{2t})}\), such that \(\Tilde{\rho}_{\vec{y}} := \frac{1}{2^n} \sum_{A \in \mathcal{P}_d^n(2t)} y(A) A\) is a valid degree-\(2t\) pseudo-density matrix. We note that \(|\vec{y}|\) is now polynomial in \(n\) and \(d\), but exponential in the SoS degree, \(t\). Equivalently, we can look at the bilinear form \(\Tilde{M}: \mathcal{L}^{(t)}(\H_d^{\otimes n}) \times \mathcal{L}^{(t)}(\H_d^{\otimes n}) \ra \C\), defined below, that will become our pseudo moment matrix.
\begin{equation} \label{SDP_mat_d}
\Tilde{M}(A, B) := \tr(\Tilde{\rho} A^* B)
\end{equation}
Which, when restricted over the subspace \(\mathcal{L}^{(t)}(\H_d^{\otimes n})\), is PSD and is polynomial in the basis of \(\mathcal{P}_d^n(2t)\).

Finally, we give the following lemma.

\begin{lemma}[degree-one terms]\label{degree_one_terms_lemma}
    Let \(\{h_\alpha\}_\alpha\) be a set of local Hamiltonians invariant under conjugation of local unitaries. For any (pseudo)-density matrix \(\rho = \frac{1}{2^n}\sum_{A \in \mathcal{P}_d^n} y(A) A\), there is another density matrix \(\rho' = \frac{1}{2^n}\sum_{A \in \mathcal{P}_d^n} y'(A) A\) such that that \(y'(A) = 0\) for all \(A \in \mathcal{P}_d^n(1) \setminus \{I\}\) being a degree one basis vector, which achieves the same expected energy, \(\E_\alpha \left[\tr(\rho h_\alpha)\right] = \E_\alpha \left[\tr(\rho' h_\alpha)\right]\).
\end{lemma}

We prove this lemma in \cref{appendix_odd_deg_terms}. This allows us to assume without loss of generality that there are no degree-one components in \(\rho\) and, in turn, will let us simplify the SDP and upper bound its value. Namely, as a result, we have the following property of the bilinear form, \(M\), using \eqref{gell-mann_prod_property}.
\begin{equation}
    M(\Lambda^a_u,\Lambda^b_u) = \tr(\rho \Lambda^a_u \Lambda^b_u) = \frac{2}{d}\delta_{ab} + \tr(\rho \mathcal{O}) = \frac{2}{d}\delta_{ab}
\end{equation}
Where \(\mathcal{O}\) is some complex combination of degree-one terms, and thus, by \cref{degree_one_terms_lemma}, their components onto \(\rho\) are zero. As a direct result, this allows us to view the degree-two ncSoS SDP as a real SDP (combined with the fact that \(M(\Lambda^a_u,\Lambda^b_v) = M(\Lambda^b_v,\Lambda^a_u)\) for \(u \neq v\)). However, this is not as easy for higher degree terms as \cref{degree_one_terms_lemma} says nothing about degree-three terms or odd degree terms in general. In fact, as we discuss in \cref{appendix_odd_deg_terms}, odd degree terms cannot, in general, can not be ignored.

\subsection{The SDP}

We finish our a priori analysis of the SDP by writing it as a vector program. To do this we use the fact that \(M \succcurlyeq 0\) and thus can be represented as a Gram matrix of the set of vectors denoted by \(\{\ket{I}\} \cup \{\ket{\Lambda^a_u}\ |\ \forall a \in [d^2 - 1],\ u \in V\} \subset \R^{n(d^2-1)+1}\). Alternatively, since the identity entries in \(M\) are constant, we consider the sub-matrix with Gram matrix of the vectors denoted by \(\{\ket{\Lambda^a_u}\ |\ \forall a \in [d^2 - 1],\ u \in V\}\). Here, we abuse the bra-ket notation to denote the SDP vectors, \(\ket{\Lambda^a_u} \in \R^{n(d^2-1)}\). The SDP vectors are such that \(\braket{\Lambda^a_u | \Lambda^b_v} = M(\Lambda^a_u, \Lambda^b_v)\). It is clearly a relaxation of \cref{basic_2LHP_def} as was observed in \eqref{pdensity_hierarchy}.

With that, we define our SDP. Our SDP mirrors the classical notion of the basic/simple SDP \cite{raghavendra_optimal_2008,raghavendra_approximating_2009}, which optimizes over valid probability distributions over the local assignments to the payoff functions and uses inner products to enforce consistency between them. For general \(k\)-local Hamiltonian problems, we would extend the SDP to enforce all \(k\)-body moments to be true density matrices. We note that similar SDPs have been used previously \cite{barthel_solving_2012,brandao_product-state_2016,parekh_application_2021,parekh_beating_2021}.

\begin{definition}\label{basicvecprog_def}[The SDP] Given a graph \(G=(V,E,w)\), and 2-local Hamiltonians, invariant under conjugation of local unitaries, \(h_{uv} = \sum_{a,b=0}^{d^2-1} c_{uv}^{ab} \Lambda^a_u  \Lambda^b_v\), for all \((u,v)\in E\), the value of the SDP is given by (for which we denote its value by \(\text{SDP}(G)\))
    \begin{maxi!}|l|
    {\substack{\rho_{uv} \in \mathcal{D}(\H_d^{\otimes 2})\\\forall u<v \in V}}{\E_{(u,v) \sim E} \left[\tr(\rho_{uv} h_{uv})\right]}{\label{basicvecprog}}{\label{basicvecprog_obj}}
    \addConstraint{\braket{\Lambda^a_u | \Lambda^b_u}}{= \frac{2}{d} \delta_{ab}\label{basicvecprog_const_square}}{\forall u \in V,\  a,b \in [d^2-1]}
    \addConstraint{\braket{\Lambda^a_u | \Lambda^b_v}}{= \tr(\rho_{uv} \Lambda^a \otimes \Lambda^b)\label{basicvecprog_loc_moment_and_M}\quad}{\forall u<v \in V,\ \forall a,b \in [d^2-1]}
    \addConstraint{\tr(\rho_{uv} \Lambda^a \otimes I)}{= 0\label{basicvecprog_loc_moment_degree_1_terms}\quad}{\forall u<v \in V,\ \forall a \in [d^2-1]}
    \addConstraint{\tr(\rho_{uv} I \otimes \Lambda^a)}{= 0\label{basicvecprog_loc_moment_degree_1_terms2}\quad}{\forall u<v \in V,\ \forall a \in [d^2-1]}
    \addConstraint{\ket{A}}{\in \R^{(d^2-1)n}\label{basicvecprog_SDP_vecs}}{\forall A \in \mathcal{P}_d^n(1) \setminus \{I\}}
    \end{maxi!}
\end{definition}

Again, note that for constraints \eqref{basicvecprog_const_square} we don't index by basis elements of \(\mathcal{P}_d^n(1)\) explicitly and instead implicitly give the constraint that \(\braket{\Lambda^a_u \otimes \Lambda^0_{V \setminus \{u\}} | \Lambda^a_u  \otimes \Lambda^0_{V \setminus \{u\}}} = \left(\frac{2}{d}\right)^n \delta_{ab}\).

Next, we discuss the reason for including the SDP variables, \(\rho_{uv}\), to represent the two-body marginals. While not the case in previous work on \textQMCProb, we note that for \(d \geq 3\), the SDP that is equivalent to the level-two ncSoS (i.e., \eqref{basicvecprog} without \(\rho_{uv}\) and constraints  \eqref{basicvecprog_loc_moment_and_M} through \eqref{basicvecprog_loc_moment_degree_1_terms2} and with the objective in terms of inner products of the SDP vectors) overshoots even the trivial upper bound of \(\QMCdProb(G) \leq 1\) for some simple graphs. As an example, this is the case for complete graphs of size \(1 < n < d\) (when \(d \geq 3\)). This is shown in \cref{clique_val_prop} (the proof of which is given in \cref{appendix_proof_clique_val_prop}).

\begin{proposition}\label{clique_val_prop}
    For \textQMCdProb with interaction graph, \(K_n\), being the unweighted complete graph on \(n\) vertices, the level-two ncSoS gets a value of \(\frac{(d-1) (d+n)}{2 d (n-1)}\).
\end{proposition}

It is for this reason that we add additional constraints to correct for this (i.e., \eqref{basicvecprog_loc_moment_and_M} to \eqref{basicvecprog_loc_moment_degree_1_terms2}). That is, our SDP is equivalent to the 2nd level of the ncSoS hierarchy in addition to enforcing that all two body moments are true density matrices. This is clearly still a relaxation of \cref{basic_2LHP_def} as a true density matrix has all of its two body moments being true density matrices. 

Moreover, we note that enforcing two body moments to be true density matrices doesn't change the number of variables in \(\vec{y}\), but rather just adds the additional constraints from the level 4 ncSoS that act on the existing variables. That is to say, taking the pseudo-density framework, for every distinct pair \(u\neq v \in [n]\), we enforce that \(\rho_{uv} := \tr_{[n]\setminus\{i,v\}}(\Tilde{\rho}) \succcurlyeq 0\). These new constraints allow the SDP to achieve the aforementioned trivial upper bound.

\begin{proposition}\label{prop_sdp_trivial_up_bound}
    For \textQMCdProb with some interaction graph \(G = (V, E)\), our SDP gets a value of \(\text{SDP}_{\text{QMC}_d}(G) \leq 1\).
\end{proposition}
\begin{proof}
    Let \(\{\rho_{uv} \in \mathcal{D}(\H_d^{\otimes 2})\}\) be the local density matrices and \(\{\ket{A}\ |\ A \in \mathcal{P}_d^n(1) \setminus \{I\}\}\) be the SDP vectors for a solution to \cref{basicvecprog_def}. We can use the fact that \(h_{uv} = I - (P_{\text{sym}})_{uv}\), being the \textQMCdProb edge interations as defined in \cref{QMCd_edge_interaction_def}, and that \(P_{\text{sym}}^2 = P_{\text{sym}}\) to get that \(\tr(\rho_{uv} P_{\text{sym}}) = \tr(\rho_{uv} P^2_{\text{sym}}) \geq 0\). Then
    \begin{align*}
        \E_{(u,v) \sim E} \left[\tr(\rho_{uv} h_{uv})\right] &= \E_{(u,v) \sim E} \left[\tr(\rho_{uv} (I - P_{\text{sym}}))\right] \\
        &= 1 - \E_{(u,v) \sim E}[\tr(\rho_{uv} P_{\text{sym}})] \\
        &\leq 1 \\
    \end{align*}
\end{proof}

Additionally, we have the following (We omit the proof as it is virtually identical to those given in previous papers \cite{parekh_application_2021,parekh_beating_2021}).

\begin{lemma}
    The vector program of \cref{basicvecprog_def} is an efficiently computable semidefinite program that provides an upper bound on \cref{basic_2LHP_def}.
\end{lemma}

Lastly, we note that enforcing all two-body moments to be valid density matrices implies useful properties. Namely, that \(\tr(\rho_{uv} P_{\text{sym}}) \geq 0\), which implies the following, by \cref{prop_Psym_is_SUd_heis_plus_indent_term} and that \(P_{\text{sym}}\) is a projector.

\begin{align}
    0 &\leq \tr(\rho_{uv} P_{\text{sym}}) \nonumber\\
    0 &\leq \frac{1}{2}\left(\frac{d+1}{d}\right) + \frac{1}{4} \sum_{a=1}^{d^2-1} \tr(\rho_{uv} \Lambda^a \otimes \Lambda^a) \nonumber\\
    \sum_{a=1}^{d^2-1} \braket{\Lambda^a_u|\Lambda^a_v} &\geq -2\left(\frac{d+1}{d}\right)\label{quantum_frieze_jerrum_style_constraint}
\end{align}


\noindent This can be seen to mirror the Frieze-Jerrum SDP constraint that \(\braket{u|v} \geq -\frac{1}{d-1}\), for SDP vectors \(\ket{u}\) and \(\ket{v}\) \cite{frieze_improved_1997}. Going one step further, we define the following vectors for each vertex, \(u \in V\), which are the concatenation of all SDP vectors associated with the vertex, \(u\), then normalized. We refer to these vectors as the \emph{stacked SDP vectors}.
\begin{equation}\label{stacked_sdp_vec}
    \ket{u}: = \frac{1}{\sqrt{\frac{2}{d}(d^2-1)}} \bigoplus_{a = 1}^{d^2-1} \ket{\Lambda^a_u} \in \R^{n(d^2-1)^2}
\end{equation}
The normalization is used to ensure that the stacked vectors are unit vectors, and it comes from the fact that
\begin{equation}\label{normalization_for_sdp_vecs}
    \sum_{a=1}^{d^2-1} \braket{\Lambda^a_u|\Lambda^a_u} = \frac{2}{d}(d^2-1)
\end{equation}

Using these stacked SDP vectors, we can reexamine \eqref{quantum_frieze_jerrum_style_constraint}, which now becomes exactly the Frieze-Jerrum SDP constraint, \(\braket{u|v} \geq -\frac{1}{d-1}\) \cite{frieze_improved_1997}. We note that this does not require us to use \cref{degree_one_terms_lemma} as \eqref{gell-mann_prod_property} gives only diagonal generalize Gell-Mann matrices, which cancel out as per the anticommutation constants \cite{bossion_general_2021}.

\subsection{The Product State SDP}

As was observed in \cite{hwang_unique_2022} for \textQMCProb, we can also define an SDP relaxation for the problem of optimizing \textQMCdProb over pure product-state solutions (\cref{def_pure_prod_prob} or equivalently \cref{prop_pure_prod_prob_rewrite}) by relaxing the size of the vectors we maximize over. We also enforce the constraint \(\braket{u|v} \geq -\frac{1}{d-1}\) on the SDP vectors as it is true of the vectors in \(\Omega_d^{\textsc{ext}}\) by \cref{bloch_vec_angle_prop}. All in all, we get the following SDP.

\begin{definition}[The Product State SDP]\label{prodsdpprog_def} For a \textQMCdProb instance with interaction graph \(G=(V,E,w)\), the value of the product state SDP is given by (for which we denote its value by \(\text{ProdSDP}_{\text{QMC}_d}(G)\))
    \begin{maxi!}|l|
        {}{\frac{1}{2}\left(\frac{d-1}{d}\right)\E_{(u,v) \sim E} \left[1-\braket{u|v}\right]}{\label{prodsdpprog}}{\label{prodsdpprog_obj}}
        \addConstraint{\braket{u|v}}{\geq -\frac{1}{d-1}\label{prodsdpprog_FJ_const}\quad}{\forall u,v \in V}
        \addConstraint{\ket{u}}{\in S^{n-1}\label{prodsdpprog_SDP_vecs}}{\forall u \in V}
    \end{maxi!}

    \begin{remark}
        This is exactly the Frieze-Jerrum SDP \cite{frieze_improved_1997} with an additional factor of \(\frac{1}{2}\).
    \end{remark}
\end{definition}

\section{Rounding}\label{section_rounding}

In this section, we prove \cref{thm_1,thm_2}. To do this, we give a brief overview of the projection rounding technique \cite{briet_positive_2010} used to round to mixed product states of qudits in much the same way that Gharibian and Parekh rounded to pure product state of single qubits \cite{gharibian_almost_2019}. For this paper, the end goal is to get a mixed product state solution. As was discussed in \cref{prod_state_section}, this can be represented efficiently by Bloch vectors for each vertex \(u \in V\) in the graph. So using the SDP defined in \cref{basicvecprog_def}, for each vertex, \(u \in V\), we will round the stacked SDP vector \(\ket{u}\), given in \eqref{stacked_sdp_vec}, to a Bloch vector \(\vec{b}_u\) (or in the case of using \cref{prodsdpprog_def}, we round using the SDP vectors directly).

Projection rounding, first introduced by Bri\"et, de Olivera Filho, and Vallentin \cite{briet_positive_2010} can be seen as a generalization of Halfspace rounding introduced introduced by Goemans and Williamson \cite{goemans_improved_1995}. The algorithm goes as follows. 

\begin{algorithm}[\textQMCdProb Rounding Algorithm]\label{algo_proj_rounding} \emph{Input:} \(\ket{u} \in S^{\ell-1}\) for each \(u \in V\) (where \(\ell = (d^2-1)^2 n\) or \(\ell = n\))
    \begin{enumerate}
        \item Pick a random matrix with i.i.d. standard Gaussian entries, \(\mathbf{Z} \sim \mathcal{N}(0,1)^{(d^2-1) \times \ell}\). 
        \item \emph{Output:} \(\vec{b}_u := \mathbf{Z} \ket{u}/\|\mathbf{Z} \ket{u}\|\) for each \(u \in V\).
    \end{enumerate}
\end{algorithm}

Bri\"et, de Olivera Filho, and Vallentin gave the following analytical tool to analyze the expectation of the inner product between two rounded vectors.

\begin{lemma}[Lemma 2.1 from \cite{briet_grothendieck_2014}]
    Let \(-1 \leq \gamma \leq 1\), and let \(\vec{u}\) and \(\vec{v}\) be two \(n\)-dimensional unit vectors such that \(\braket{\vec{u},\vec{v}} = \gamma\). Let \(\mathbf{Z} \in \R^{k \times n}\) be a random matrix with entries chosen from \(kn\) i.i.d. standard Gaussians, \(\mathcal{N}(0, 1)\). Then
    \[F^*(k, \gamma) := \E_\mathbf{Z} \left\langle\frac{\mathbf{Z} \vec{u}}{\left\|\mathbf{Z}\vec{u}\right\|},\frac{\mathbf{Z} v}{\left\|\mathbf{Z} v\right\|}\right\rangle = \frac{2 \gamma}{k}\left(\frac{\Gamma((k+1)/2)}{\Gamma(k/2)}\right)^2 \, _2 F_1\left(\frac{1}{2}, \frac{1}{2}; \frac{k}{2} + 1; \gamma^2\right)\]
    Where \(\, _2 F_1(\cdot, \cdot; \cdot; \cdot)\) is the Gaussian hypergeometric function.
\end{lemma}

Our full algorithm is to solve the SDP given in \cref{basicvecprog_def} and then use the stacked vectors, given by \eqref{stacked_sdp_vec}, as input to \cref{algo_proj_rounding}. For finding the approximation ratio, \(\alpha_d\), it is enough to consider the worst-case edge. That is we consider an \(\alpha_d\) such that for all \((u,v) \in E\) we have that

\begin{equation}\label{approx_ratio1}
\E_{\mathbf{Z}} \left[ \frac{1}{2}\left(\frac{d-1}{d}\right)\left(1-\frac{\braket{\vec{b}_u, \vec{b}_v}}{(d-1)^2}\right)\right] \geq \alpha_d \frac{1}{2}\left(\frac{d-1}{d}\right)\left(1-(d+1)\braket{u|v}\right)
\end{equation}

Here, the left hand side comes from \cref{prop_rewite_mixed_state_and_upper_bound} and the right hand side comes from applying \eqref{stacked_sdp_vec} to the SDP objective, \eqref{basicvecprog_obj}, of \cref{basicvecprog_def}; namely,

\begin{align*}
    \tr(\rho_{uv} h_{uv}) &= \frac{1}{2}\left(\frac{d-1}{d}\right) - \frac{1}{4} \sum_{a = 1}^{d^2-1}\braket{\Lambda^a_u|\Lambda^a_v} \\
    &= \frac{1}{2}\left(\frac{d-1}{d}\right) - \frac{1}{2d}(d^2-1) \braket{u|v} \\
    &= \frac{1}{2}\left(\frac{d-1}{d}\right)\left(1 - (d+1) \braket{u|v}\right) \\
\end{align*}

We calculate this worst case value by minimizing over \(\gamma = \braket{u|v} \geq -\frac{1}{d-1}\), as per the observation give in \eqref{quantum_frieze_jerrum_style_constraint}. We also have that \(\gamma \leq \frac{1}{d+1}\) as otherwise the bound on the optimal energy would be negative, which is not possible as \(h_{uv}\) is PSD. This bound can be derived from \(\tr(\rho_{uv} h_{uv}) \geq 0\) in much the same way we derived \eqref{quantum_frieze_jerrum_style_constraint}. This then gives us the following expression for \(\alpha_d\).

\begin{equation}\label{approx_ratio2}
    \alpha_d = \min_{-\frac{1}{d-1} \leq \gamma < \frac{1}{d+1}} \frac{1 - \frac{F^*(d^2-1, \gamma)}{(d-1)^2}}{1 - (d+1)\gamma}
\end{equation}

We note that for \(d=2\), this is exactly the qubit case and thus gives the same expression and approximate ratio. Moreover, as was done for the analysis of \textQMCProb, we can then numerically analyze \eqref{approx_ratio2}, which gives the following (shown in \cref{tab:approx_ratios}). For completeness, we also give the optimal approx ratios for mixed states with \(\tr(\rho^2) = \frac{1}{d-1}\) (which is given by \eqref{opt_ratio_for_mix_prod_state} in the proof of \cref{prop_rewite_mixed_state_and_upper_bound}). We show that the values of \(\gamma_d\) are exact (for \(d \geq 3\)) in \cref{lemma_bad_angle}.

\begin{table}[ht!]
    \centering
    \begin{tabular}{c|c|c|l}
        \(d\) & \(\alpha_d\) & \(\gamma_d\) & Opt Mixed Prod State Ratio \\
        \hline
        2 & 0.498767 & \(-0.9659\) & \(\,\frac{1}{2} = 0.5\) \\
        3 & 0.372996 & \(-1/2\) & \(\frac{5}{12} \approx 0.416667\)  \\
        4 & 0.388478 & \(-\frac{1}{3}\) & \(\frac{5}{12} \approx 0.416667\)  \\
        5 & 0.406129 & \(-\frac{1}{4}\) & \(\frac{17}{40} = 0.425\) \\
        10 & 0.450614 & \(-\frac{1}{9}\) & \(\frac{41}{90} \approx 0.455556\) \\
        100 & 0.495001 & \(-\frac{1}{99}\) & \(\frac{4901}{9900} \approx 0.495051\) \\
    \end{tabular}
    \caption{The numerically calculated approximation ratios, the value of \(\gamma_d\) that minimized the expression (known as the ``bad angle"), and optimal mixed product state approximation ratios with \(\tr(\rho^2) = \frac{1}{d-1}\)  for different values of \(d\) (from \eqref{opt_ratio_for_mix_prod_state}). For \(d \geq 3\), these values of \(\gamma_d\) are exact as stated in \cref{lemma_bad_angle}.}
    \label{tab:approx_ratios}
\end{table}

Similarly, for the approximation to the optimal product-state solution (\cref{def_pure_prod_prob}), our full algorithm is to solve the SDP given in \cref{prodsdpprog_def} and then use the SDP vectors as input to \cref{algo_proj_rounding}. The approximation ratios, \(\beta_d\), are then such that.

\begin{equation}\label{approx_ratio3}
\E_{\mathbf{Z}} \left[ \frac{1}{2}\left(\frac{d-1}{d}\right)\left(1-\frac{\braket{\vec{b}_u, \vec{b}_v}}{(d-1)^2}\right)\right] \geq \beta_d \frac{1}{2}\left(\frac{d-1}{d}\right)\left(1-\braket{u|v}\right)
\end{equation}
Here, the left hand side again comes from \cref{prop_rewite_mixed_state_and_upper_bound} and the right hand side comes directly from the SDP objective, \eqref{basicvecprog_obj}, of \cref{basicvecprog_def}. Finally,
\begin{equation}\label{approx_ratio4}
    \beta_d = \min_{-\frac{1}{d-1} \leq \gamma < 1} \frac{1 - \frac{F^*(d^2-1, \gamma)}{(d-1)^2}}{1 - \gamma}
\end{equation}
where the constraint \(\gamma = \braket{u|v} \geq -\frac{1}{d-1}\) comes directly from the product-state SDP, \eqref{prodsdpprog_FJ_const}. We hold off on solving for these values numerically as they will be determined exactly by \cref{thm_2}.

\begin{lemma}[The Bad Angle For \(\alpha_d\)]\label{lemma_bad_angle}
    For \(d \geq 3\), the value of \(-\frac{1}{d-1} \leq \gamma < \frac{1}{d+1}\) that achieves the approximation ratio in \eqref{approx_ratio2} is \(\gamma = -\frac{1}{d-1}\).
\end{lemma}
\begin{proof}
    In order to prove this, we first prove that for integers \(d \geq 3\) and for \(\gamma\) in the range \(-\frac{1}{d-1} \leq \gamma < \frac{1}{d+1}\), we have that the partial derivative of \eqref{approx_ratio2} with respect to \(\gamma\) is non-negative:
    \[\frac{\partial}{\partial \gamma} \left\{\frac{1 - \frac{F^*(d^2-1, \gamma)}{(d-1)^2}}{1 - (d+1)\gamma}\right\} \geq 0\]
    It then follows that \eqref{approx_ratio2} is minimized when \(\gamma = -\frac{1}{d-1}\) takes the smallest value in this range.

    First, recall the series expansion definition of the Gaussian hypergeometric function:
    \begin{equation}\label{hypgeo_func_def}
        \, _2 F_1\left(a, b; c; z\right) := \sum_{n=0}^{\infty} \frac{(a)_n (b)_n}{(c)_n} \frac{z^n}{n!} = 1 + \frac{ab}{c}\frac{z}{1}+ \frac{a(a+1)b(b+1)}{c(c+1)}\frac{z^2}{2} + \cdots
    \end{equation}
    where \((x)_n = \prod_{k=0}^{n-1} (x+k)\) denotes the rising factorial. Next, we note that the derivative of the Gaussian hypergeometric function with respect \(z\) is the following.
    \begin{equation}\label{hypgeo_func_deriv}
        \frac{d z}{d} \, _2 F_1\left(a, b; c; z\right) = \frac{ab}{c} \, _2 F_1\left(a + 1, b + 1; c + 1; z\right)
    \end{equation}
    As an immediate result of this, we know that for \(a,b,c > 0\), the derivative of the Gaussian hypergeometric function is positive for \(z > 0\). Moreover, all of its derivatives are positive. In other words, the Gaussian hypergeometric function for \(a,b,c \geq 0\) is concave up when \(z > 0\). 
    
    Using \eqref{hypgeo_func_deriv}, we have that
    \begin{align*}
        \frac{\partial}{\partial \gamma}\left\{F^*\left(d^2-1,\gamma \right)\right\} = \frac{2}{d^2-1}\left(\frac{\Gamma\left(\frac{d^2}{2}\right)}{\Gamma\left(\frac{1}{2}\left(d^2-1\right)\right)}\right)^2 \frac{\gamma^2 \, _2F_1\left(\frac{3}{2},\frac{3}{2};\frac{1}{2} \left(d^2+3\right);\gamma ^2\right)+\left(d^2+1\right) \, _2F_1\left(\frac{1}{2},\frac{1}{2};\frac{1}{2} \left(d^2+1\right);\gamma ^2\right)}{\left(d^2+1\right)}
    \end{align*}
    Then,
    \begin{align}
        \frac{\partial}{\partial \gamma} \left\{\frac{1 - \frac{F^*(d^2-1, \gamma)}{(d-1)^2}}{1 - (d+1)\gamma}\right\} &= \frac{(d+1) (d-1)^2-(d+1) F^*\left(d^2-1,\gamma \right) - (1 - (d+1)\gamma) \frac{\partial}{\partial \gamma}\left\{F^*\left(d^2-1,\gamma \right)\right\}}{(d-1)^2 (1 - (d+1)\gamma)^2} \nonumber\\
        \begin{split}
            &= \frac{d+1}{(1 - (d + 1) \gamma)^2}-
            \frac{2}{d^2-1}\left(\frac{\Gamma\left(\frac{d^2}{2}\right)}{\Gamma\left(\frac{1}{2}\left(d^2-1\right)\right)}\right)^2 \left(\frac{1}{(d-1)^2 \left(d^2+1\right) (1 - (d + 1) \gamma)^2}\right) \\
            &\ \ \ \ \ \ \ \ \ \ \ \ \ \ \ \ \cdot \Bigg(\left(d^2+1\right) \, _2F_1\left(\frac{1}{2},\frac{1}{2};\frac{1}{2} \left(d^2+1\right);\gamma ^2\right)\\
            &\ \ \ \ \ \ \ \ \ \ \ \ \ \ \ \ \ \ \ \ + \gamma ^2 (1- (d + 1) \gamma) \, _2F_1\left(\frac{3}{2},\frac{3}{2};\frac{1}{2} \left(d^2+3\right);\gamma ^2\right)\Bigg)\label{left_off_eq_1}
        \end{split}
    \end{align}
    For which, we direct our attention to the following, which we denoted by \(g\).
    \[g(d,\gamma) := \frac{1}{d^2+1}\left(\left(d^2+1\right) \, _2F_1\left(\frac{1}{2},\frac{1}{2};\frac{1}{2} \left(d^2+1\right);\gamma ^2\right) + \gamma^2 (1- (d + 1) \gamma) \, _2F_1\left(\frac{3}{2},\frac{3}{2};\frac{1}{2} \left(d^2+3\right);\gamma ^2\right)\right)\]
    We seek to upper bound this in the range \(-1 \leq \gamma \leq \frac{1}{d+1}\). We can observe that the partial derivative on \(\gamma\) has two roots at \(\gamma = 0, \frac{1}{d+1}\) for all \(d\), as
    \[\frac{\partial}{\partial \gamma} \left\{g(d,\gamma)\right\} = \frac{3 \gamma  (1-(d+1)\gamma) \left(3 \gamma ^2 \, _2F_1\left(\frac{5}{2},\frac{5}{2};\frac{1}{2} \left(d^2+5\right);\gamma ^2\right)+\left(d^2+3\right) \, _2F_1\left(\frac{3}{2},\frac{3}{2};\frac{1}{2} \left(d^2+3\right);\gamma ^2\right)\right)}{\left(d^2+1\right) \left(d^2+3\right)}\]
    So to bound the maximum value in this range, we consider the maximum value at these points combined with the point at \(\gamma = -1\). Clearly, \(g(d,0) = 1\) for all \(d\) as is evident by the series expansion, \eqref{hypgeo_func_def}. Additionally, \(g(d, -1) > g(d, -\frac{1}{d+1}) > g(d, \frac{1}{d+1})\) as the Gaussian hypergeometric function is concave up for \(\gamma^2 > 0\) and the term \(\gamma^2 (1- (d + 1) \gamma)\) is maximized when \(\gamma = -1\). Therefore, we are left to maximize \(g(d,-1)\) over \(d \geq 3\).
    \[g(d,-1) = \, _2F_1\left(\frac{1}{2},\frac{1}{2};\frac{1}{2} \left(d^2+1\right);1\right)+\frac{(d+2)}{d^2+1} \, _2F_1\left(\frac{3}{2},\frac{3}{2};\frac{1}{2} \left(d^2+3\right);1\right)\]
    By the series expansion, we have that \(\, _2F_1\left(\frac{a}{2},\frac{b}{2};\frac{1}{2} \left(d^2+c\right);1\right) = 1 + \Theta(\frac{1}{d^2})\) and is decreasing and thus \(g(d,-1) = 1 + \Theta(\frac{1}{d})\). Thus \(g(d,\gamma) \leq g(d,-1) \leq g(3,-1) < 2\) and is decreasing (for \(d \geq 3\)).
    
    Then, we can rearrange \eqref{left_off_eq_1} and use the upper bound \(g(d,\gamma) < 2\), to get
    \begin{align*}
        \frac{\partial}{\partial \gamma} \left\{\frac{1 - \frac{F^*(d^2-1, \gamma)}{(d-1)^2}}{1 - (d+1)\gamma}\right\} 
        &> \frac{d+1}{(1 - (d + 1) \gamma)^2}-
        \frac{4}{d^2-1}\left(\frac{\Gamma\left(\frac{d^2}{2}\right)}{\Gamma\left(\frac{1}{2}\left(d^2-1\right)\right)}\right)^2 \left(\frac{1}{(d-1)^2 \left(d^2+1\right) (1 - (d + 1) \gamma)^2}\right) \\
        &= \frac{d+1}{(1 - (d + 1) \gamma)^2}\left(1-
        \frac{4}{d^2-1}\left(\frac{\Gamma\left(\frac{d^2}{2}\right)}{\Gamma\left(\frac{1}{2}\left(d^2-1\right)\right)}\right)^2 \left(\frac{1}{(d-1)^2 \left(d^2+1\right) (d+1)}\right)\right)
    \end{align*}
    Which we want to show is positive. To do this, we show that both terms in the product are positive. First, we can observe that \(\frac{d+1}{(1 - (d + 1) \gamma)^2} > 0\) for all \(d \geq 3\) and \(\gamma < \frac{1}{d+1}\). Next, we can use Gautschi's inequality, which states that for \(x > 1\), \(\frac{\Gamma(x)}{\Gamma(x-\frac{1}{2})} < \sqrt{x}\), to get the following for \(d \geq 3\).
    \[\frac{4}{d^2-1}\left(\frac{\Gamma\left(\frac{d^2}{2}\right)}{\Gamma\left(\frac{1}{2}\left(d^2-1\right)\right)}\right)^2 < \frac{4}{d^2-1} \frac{d^2}{2} \leq \left(\frac{3}{2}\right)^2\] 
    Thus, the whole term is positive for \(d \geq 3\), and so the product is also positive for all \(d \geq 3\).
\end{proof}

\begin{lemma}[The Bad Angle for \(\beta_d\)]\label{lemma_bad_angle2}
    For \(d \geq 3\), the value of \(-\frac{1}{d-1} \leq \gamma < 1\) that achieves the approximation ratio in \eqref{approx_ratio4} is \(\gamma = -\frac{1}{d-1}\).
\end{lemma}
\begin{proof}[Proof (sketch)]
    This follows by the same logic as the proof for \cref{lemma_bad_angle}. We note that
    \begin{align*}
        \frac{\partial}{\partial \gamma} \left\{\frac{1 - \frac{F^*(d^2-1, \gamma)}{(d-1)^2}}{1 - \gamma}\right\} &= \frac{(d-1)^2 - F^*\left(d^2-1,\gamma \right) - (1 - \gamma) \frac{\partial}{\partial \gamma}\left\{F^*\left(d^2-1,\gamma \right)\right\}}{(d-1)^2 (1 - \gamma)^2} \\
        \begin{split}
            &= \frac{1}{(1-\gamma)^2} \Biggl(1 - \frac{2}{d^2-1}\left(\frac{\Gamma\left(\frac{d^2}{2}\right)}{\Gamma\left(\frac{1}{2}\left(d^2-1\right)\right)}\right)^2 \left(\frac{1}{(d-1)^2(d^2+1)}\right) \\
            &\ \ \ \ \ \ \ \ \ \ \ \ \ \ \ \ \ \ \ \ \ \ \ \ \cdot \biggl(\left(d^2+1\right) \, _2F_1\left(\frac{1}{2},\frac{1}{2};\frac{1}{2} \left(d^2+1\right);\gamma ^2\right)\\
            &\ \ \ \ \ \ \ \ \ \ \ \ \ \ \ \ \ \ \ \ \ \ \ \ \ \ \ \ + \gamma^2 (1 - \gamma) \, _2F_1\left(\frac{3}{2},\frac{3}{2};\frac{1}{2} \left(d^2+3\right);\gamma ^2\right)\biggr)\Biggr)
        \end{split}
    \end{align*}
    For which, the term with the two Gaussian hypergeometric functions can be show to be less than \(\frac{3}{2}\) for \(d \geq 3\).
\end{proof}

As a result of \cref{lemma_bad_angle,lemma_bad_angle2}, we can plug the minimizers into \eqref{approx_ratio2} and \eqref{approx_ratio4} to get the following.
\begin{equation}\label{approx_ratio5}
    \alpha_d = \frac{1}{2}\left(\frac{d-1}{d}\right)\left(1 - \frac{F^*(d^2-1, -\frac{1}{d-1})}{(d-1)^2}\right),\ \ \ \ \beta_d = \left(\frac{d-1}{d}\right)\left(1 - \frac{F^*(d^2-1, -\frac{1}{d-1})}{(d-1)^2}\right)\ \ \ \ \text{for } d \geq 3
\end{equation}

\subsection{Proof of \texorpdfstring{\cref{thm_1}}{Theorem \ref{thm_1}} and \texorpdfstring{\cref{thm_2}}{Theorem \ref{thm_2}}}

Finally, we can prove \cref{thm_1}.

\begin{theorem}[Restatement of \cref{thm_1}]\label{thm_1_restate} There exists an efficient approximation algorithm for \textQMCdProb that admits an \(\alpha_d\)-approximation, where the constants \(\alpha_d\) (for \(d \geq 2\)) satisfy, 
\begin{enumerate}
    \item[(i)] \(\alpha_d \geq \frac{1}{2}\left(1-\frac{1}{d}\right)\)
    \item[(ii)] \(\alpha_d - \frac{1}{2}\left(1-\frac{1}{d}\right) \sim \frac{1}{2 d^3}\)
    \item[(iii)] \(\alpha_2 \geq 0.498767,\alpha_3 \geq 0.372995,\alpha_4 \geq 0.388478,\alpha_5 \geq 0.406128,\alpha_{10} \geq 0.450614,\alpha_{100} \geq 0.4950005\)
\end{enumerate}
\end{theorem}

\begin{proof}
We prove \cref{thm_1}, \cref{thm_1_pt_1,thm_1_pt_2} for \(d \geq 3\). The \(d=2\) case follows from \cref{thm_1_pt_3}. First, we use \eqref{approx_ratio5}, which is a result of \cref{lemma_bad_angle} and simplify:
\begin{align}
    \alpha_d%
    &= \frac{1}{2}\left(1-\frac{1}{d}\right) - \frac{F^*\left(d^2-1, -\frac{1}{d-1}\right)}{2d(d-1)} \nonumber\\
    &=\frac{1}{2}\left(1-\frac{1}{d}\right)+\frac{1}{(d-1)^3 d (d+1)}\left(\frac{\Gamma\left(\frac{d^2}{2}\right)}{\Gamma\left(\frac{d^2-1}{2}\right)}\right)^2 \, _2 F_1\left(\frac{1}{2}, \frac{1}{2}; \frac{1}{2}(d^2 + 1); \frac{1}{(d-1)^2}\right)\label{main_thm_1_eq}
\end{align}
Using the series expansion, \eqref{hypgeo_func_def}, for the values $a,b=1/2,c=\frac1{2}(d^2+1),z=\frac{1}{(d-1)^2}$ as in \eqref{main_thm_1_eq}, we always have that (for any \(n \geq 0\))
\[\frac{\left(\frac{1}{2}\right)_n \left(\frac{1}{2}\right)_n}{\left(\frac{1}{2}(d^2+1)\right)_n} \frac{1}{(d-1)^{2n} n!} > 0\]
Then plugging in our values gives the series approximation of \(1 + O(d^{-4})\).
Furthermore, because each term in the series is positive, the \(O(d^{-4})\) term is positive. This proves \cref{thm_1}, \cref{thm_1_pt_1}.

Next, we look at the expression in \cref{thm_1}, \cref{thm_1_pt_2}, which using \eqref{main_thm_1_eq} gives the following equality. Moreover, we use the series approximation of the Gaussian hypergeometric function with the above inputs to get the asymptotic equivalence below. The series approximation, \(1 + o_d(1)\), tells us that it is asymptotically equivalent to 1.
\[\alpha_d - \frac{1}{2}\left(1-\frac{1}{d}\right) = \left(\frac{\Gamma\left(\frac{d^2}{2}\right)}{\Gamma\left(\frac{1}{2}\left(d^2-1\right)\right)}\right)^2 \left(\frac{\, _2 F_1\left(\frac{1}{2}, \frac{1}{2}; \frac{1}{2}(d^2+1); \frac{1}{(d-1)^2}\right)}{(d-1)^3 d (d+1)}\right) \sim \frac{1}{d^5} \left(\frac{\Gamma\left(\frac{d^2}{2}\right)}{\Gamma\left(\frac{1}{2}\left(d^2-1\right)\right)}\right)^2\]
Furthermore, by using Stirling's approximation/the Lanczos approximation, which says that \(\Gamma(z) \sim \sqrt{2\pi} z^{z-1/2}e^{-z}\) (as \(z \ra \infty\)), we get the following.
\[\left(\frac{\Gamma\left(\frac{d^2}{2}\right)}{\Gamma\left(\frac{1}{2}\left(d^2-1\right)\right)}\right)^2 \sim \frac{1}{2 e}\frac{(d^2)^{d^2-1}}{(d^2-1)^{d^2-2}} = \frac{1}{2 e}\frac{d^2}{\left(1 - d^{-2}\right)^{-2}\left(1 - d^{-2}\right)^{d^2}} \sim \frac{d^2}{2}\]

Finally, we can put it all together to get that
\[\alpha_d - \frac{1}{2}\left(1-\frac{1}{d}\right) \sim \frac{1}{d^5} \left(\frac{\Gamma\left(\frac{d^2}{2}\right)}{\Gamma\left(\frac{1}{2}\left(d^2-1\right)\right)}\right)^2 \sim \frac{1}{2d^3}\]
This proves \cref{thm_1}, \cref{thm_1_pt_2}.

\cref{thm_1}, \cref{thm_1_pt_3} is proven numerically and can be seen in \cref{tab:approx_ratios}.
\end{proof}

Next, we can prove \cref{thm_2}.

\begin{theorem}[Restatement of \cref{thm_2}]\label{thm_2_restate} \textQMCdProb admits an \(\beta_d\)-approximation to the optimal product-state solution with respect to the basic SDP, where the constants \(\beta_d\) satisfy, 
\begin{enumerate}[label=(\roman*)]
    \item \(\beta_d = 2 \alpha_d\) for \(d \geq 3\)
    \item \(\beta_2 \geq 0.956337\) \cite{briet_positive_2010,hwang_unique_2022}
\end{enumerate}
\end{theorem}

\begin{proof}
    First, for \cref{thm_2}, \cref{thm_2_pt_1}, this follows from \eqref{approx_ratio5}, which are results of \cref{lemma_bad_angle,lemma_bad_angle2}. Then, \cref{thm_2}, \cref{thm_2_pt_2} is proven numerically.
\end{proof}

\section{The Complete Graph As The Algorithmic Gap}\label{section_alg_gap}

In this section, we prove \cref{thm_alg_gap} by showing that our analysis for our rounding algorithm is tight. This is often done through the notion of the algorithmic gap, defined in \cref{def_alg_gap}. In particular, we show that the algorithmic gap of our algorithm matches the approximation ratio given in \cref{thm_1}.

\begin{theorem}[Restatement of \cref{thm_alg_gap}] The approximation algorithm for \textQMCdProb that rounds to mixed product states using the basic SDP has algorithmic gap \(\alpha_d\) for \(d \geq 3\).
\end{theorem}

To show that this is the algorithmic gap, it suffices to give an instance that achieves the ratio \(\alpha_d\) in expectation. This is because the analysis determining the approximation ratio can be seen as a lower bound on the algorithmic gap. The problem instance we will use will be that of the complete graph on \(d\) vertices, denoted \(K_d\). The following is a well-known result (e.g., see \cite{brandao_product-state_2016}). However, we nonetheless give the proof in \cref{appendix_proof_lemma_d_clique_energy}.

\begin{lemma}[Energy of \(K_d\)]\label{lemma_d_clique_energy}
    The energy of \textQMCdProb with the complete graph on \(d\) vertices, \(K_d\), as it's interaction graph is 1.
\end{lemma}

Next, we prove \cref{thm_alg_gap} by looking at the SDP vectors. We note that in \cref{thm_1}, we already showed that the algorithmic gap is at least \(\alpha_d\). It then suffices to give an instance, \(G\), where our rounding algorithm outputs a solution with value \(\alpha_d \text{OPT}(G)\) in expectation. We show that the compete graph on \(d\) vertices, \(K_d\), is such an instance.

\begin{proof}
    Consider the graph \(K_d\) and it's corresponding Hamiltonian \(H_{K_d}\). Then let \(M\) be an optimal SDP solution with corresponding vector program solution \(\{\ket{A}\ |\ A \in \mathcal{P}_d^d(1)\}\) to \cref{basicvecprog_def}.

    We can observe that the SDP, being a relaxation, must get a value at least \(\text{SDP}_{\text{QMC}_d}(K_d) \geq 1\) (here, we use \cref{lemma_d_clique_energy}). And, by \cref{prop_sdp_trivial_up_bound}, we know that \(\text{SDP}_{\text{QMC}_d}(K_d) \leq 1\). Moreover, this is true for each edge interaction as observed in the proof of \cref{prop_sdp_trivial_up_bound}. All in all, this means that our SDP solution for \(H_{K_d}\) has that \(\braket{u|v} = -\frac{1}{d-1}\) for each \(u < v \in V\). This is exactly the bad angle, and thus, plugging it into \cref{def_alg_gap}, we get the following, where the denominator is 1, by \cref{lemma_d_clique_energy}.
    \begin{align*}
        \text{Gap}_A(K_d) &=\E_{Z} \left[\frac{1}{2}\left(\frac{d-1}{d}\right) \E_{uv \sim E} \left(1-\frac{\braket{\vec{b}_u, \vec{b}_v}}{(d-1)^2}\right)\right] \\
        &= \frac{1}{2}\left(\frac{d-1}{d}\right) \left(1-\frac{F^*(d^2-1, -\frac{1}{d-1})}{(d-1)^2}\right) \\
        &= \alpha_d
    \end{align*}
    The last equality follows from \cref{lemma_bad_angle} for \(d \geq 3\).
\end{proof}

\section{Conclusion and Future Directions}\label{section_conclusion}

In this paper, we look at an approximation algorithm for \textQMCdProb and prove it beats random assignment. In particular, our algorithm finds a mixed product state solution. Moreover, we show that our analysis is tight by finding an algorithmic gap instance for our algorithm. However, our algorithm is not optimal as is evedent by the Frieze-Jerrum Algorithm \cite{frieze_improved_1997}. We believe that rounding directly to pure product state solutions using a clever choice for a frame could produce a better algorithm. Nonetheless, our paper makes progress on approximating LHPs over qudits. We give the following open problems.

Does there exist a pure product state rounding algorithm that achieves a better approximation ratio than ours? Gharibian and Kempe \cite{gharibian_approximation_2012} as well as Brand\~ao and Harrow \cite{brandao_product-state_2016}, for example, looked at a rounding algorithm to pure product state solution for general LHPs over qudits, but required assumptions for their graphs. Does there exist algorithms that do not require these assumptions?

Additionally, after the Gharibian-Parekh algorithm, many subsequent papers considered rounding to entangled states \cite{anshu_beyond_2020,parekh_application_2021,king_improved_2022,lee_optimizing_2022,parekh_optimal_2022} and achieved approximation ratios better than the optimal product state ratio of \(1/2\). Can similar algorithms using higher levels of the ncSoS hierarchy be used for \textQMCdProb. The major road block for this is in establishing something analogous to the ``star bound" \cite{parekh_application_2021} for qudits. This may prove to be challenging as odd degree terms can not be in general ignored. However, with recent work looking at the algebra of swap operators to solve \textQMCProb and its relaxation \cite{takahashi_su2-symmetric_2023,watts_relaxations_2023}, the question arises if similar results can be achieved for \textQMCdProb by looking at polynomial optimization in the qudit swap operators.

In the analysis of classical approximation algorithms for CSPs/GCSPs, one metric that is used to show the optimality of an SDP-based rounding algorithm is the integrality gap. If the integrality gap matches the approximation ratio, we say that the rounding algorithm is optimal for the SDP. Work done by Hwang, Neeman, Parekh, Thompson, and Wright \cite{hwang_unique_2022} showed that the integrality gap of the 2nd level ncSoS SDP for \textQMCProb, assuming a plausible conjecture in Gaussian geometry, matches the approximation ratio of the Gharibian-Parekh algorithm. It is natural to ask if we can show an integrity gap for the 2nd level ncSoS SDP that also enforces two-body moments to be valid density matrices for the problem of \textQMCdProb or even \textQMCProb.

\section*{Acknowledgements}

This work was done in part while ZJ was visiting the Simons Institute for the Theory of Computing. Additionally, the authors thank John Wright for the helpful discussions and ZJ thanks Ian Jorquera.

\printbibliography

\appendix

\section{Proofs of \texorpdfstring{\cref{prop_Psym_is_SUd_heis_plus_indent_term,QMCk_edge_interaction_gell-mann_prop}}{Propositions \ref{prop_Psym_is_SUd_heis_plus_indent_term} and \ref{QMCk_edge_interaction_gell-mann_prop}}}\label{appendix_heis_model_and_QMCd}

\begin{remark}
    In this section, we use commas in the bra-ket notation of multiple qudits to aid in readability. Namely, \(\ket{a,b} := \ket{a} \otimes \ket{b}\).
\end{remark}

\begin{proposition}[Restatement of \cref{QMCk_edge_interaction_gell-mann_prop}]\label{QMCk_edge_interaction_gell-mann_prop_appendix}
The \textQMCdProb edge interaction can be written in terms of the generalized Gell-Mann matrices in the following way.
\begin{equation}\label{QMCk_edge_interaction_gell-mann_eq_restated}
    \sum_{a< b \in [d]} \left(\frac{1}{\sqrt{2}}\ket{a,b} - \frac{1}{\sqrt{2}}\ket{b,a}\right)\left(\frac{1}{\sqrt{2}}\bra{a,b} - \frac{1}{\sqrt{2}}\bra{b,a}\right) = \frac{1}{2}\left(\frac{d-1}{d}\right) I - \frac{1}{4}\sum_{a=1}^{d^2-1}\Lambda^a \otimes \Lambda^a
\end{equation}
\end{proposition}

\begin{proof}
We refer to the lefthand side of \eqref{QMCk_edge_interaction_gell-mann_eq_restated} as $h$ and the righthand side as $g$. First, expand $h$:

\begin{equation}
    h = \sum_{a< b \in [d]} \left(\frac{1}{\sqrt{2}}\ket{a,b} - \frac{1}{\sqrt{2}}\ket{b,a}\right)\left(\frac{1}{\sqrt{2}}\bra{a,b} - \frac{1}{\sqrt{2}}\bra{b,a}\right) = \frac{1}{2}\sum_{a \neq b \in [d]} \ket{a,b}\bra{a,b} - \frac{1}{2}\sum_{a \neq b \in [d]} \ket{a,b}\bra{b,a} \label{prop2_h_expanded}
\end{equation}

\noindent The first summation says that the non-zero elements on the diagonal of $h$ are all $+1/2$ and that the non-zero elements on the off-diagonal are all $-1/2$. We will show that the coefficients of $g$ match these.

Next, re-write $g$ with respect to the different types of Gell-Mann matrices (see definition \cref{gen_gell-mann_def}):
\begin{equation}
    g =  \frac{1}{2}\left(\frac{d-1}{d}\right) I - \frac{1}{4}\sum_{a < b \in [d]}^{d^2-1} \left(\Lambda^+_{ab} \otimes \Lambda^+_{ab} +\Lambda^-_{ab} \otimes \Lambda^-_{ab}\right)- \frac{1}{4}\sum_{a = 1}^{d-1}\Lambda^d_a \otimes \Lambda^d_a
\end{equation}
Since the $\Lambda^d_a \otimes \Lambda^d_a$ terms are all diagonal, the middle summation is composed of the off-diagonal elements of $g$. For $a < b \in [d]$,

    \begin{align}
        \Lambda_{ab}^+ \otimes \Lambda_{ab}^+ + \Lambda_{ab}^- \otimes \Lambda_{ab}^- &= (\ket{a}\bra{b}+\ket{b}\bra{a}) \otimes (\ket{a}\bra{b}+\ket{b}\bra{a}) + -i(\ket{a}\bra{b}-\ket{b}\bra{a}) \otimes -i(\ket{a}\bra{b}-\ket{b}\bra{a})\nonumber 
        \\ &=2 \ket{a,b}\bra{b,a} + 2\ket{b,a}\bra{a,b} \\
        \Longrightarrow g &= \frac{1}{2}\left(\frac{d-1}{d}\right) I - \frac{1}{2}\sum_{a \neq b \in [d]}^{d^2-1} \ket{a,b}\bra{b,a} - \frac{1}{4}\sum_{a = 1}^{d-1}\Lambda^d_a \otimes \Lambda^d_a \label{prop2eq1}
    \end{align}

    \noindent \eqref{prop2eq1} shows that the off-diagonal elements of $g$ have coefficient $-1/2$, matching the terms in \eqref{prop2_h_expanded}. We now show that the non-zero diagonal elements have coefficient $+1/2$. 
    
    For simplicity, let $g' := \frac{1}{2}\left(\frac{d-1}{d}\right) I - \frac{1}{4}\sum_{a = 1}^{d-1}\Lambda^d_a \otimes \Lambda^d_a$ be the operator consisting of the diagonal elements of $g$. For $a \in [d-1]$:

    \begin{align}
        \Lambda^d_a \otimes \Lambda^d_a &= \sqrt{\frac{2}{a(a+1)}}\left( \sum_{b = 1}^a \ket{b} \bra{b} - a \ket{a+1}\bra{a+1} \right) \otimes \sqrt{\frac{2}{a(a+1)}}\left( \sum_{b' = 1}^a \ket{b'} \bra{b'} - a \ket{a+1}\bra{a+1} \right) \nonumber \\
        \begin{split}\label{prop2eq2}
            &= \frac{2}{a(a+1)}\sum_{b,b' = 1}^a \ket{b,b'}\bra{b,b'} -\frac{2}{a+1}\left(\sum_{b=1}^a\ket{a+1, b} \bra{a+1,b} + \ket{b,a+1} \bra{b,a+1} \right) \\
            &\ \ \ \ \ \ \ \ + \frac{2a}{a+1}\ket{a+1,a+1}\bra{a+1,a+1}
        \end{split}
    \end{align}

    \noindent Working with the first summation:

    \begin{align}
        \sum_{a=1}^{d-1}\frac{2}{a(a+1)}&\sum_{b,b' = 1}^a \ket{b,b'}\bra{b,b'} \nonumber\\
        = &\ \frac{2}{1 \cdot 2}\ket{11}\bra{11} + \frac{2}{2 \cdot 3} \left(\ket{11}\bra{11} + \ket{12}\bra{12} + \ket{21}\bra{21} + \ket{22}\bra{22}\right) + \cdots + \frac{2}{d(d-1)}\sum_{b,b'=1}^{d-1}\ket{b,b'}\bra{b,b'} \nonumber \\
        =&\ \left(\frac{2}{1\cdot 2} + \frac{2}{2 \cdot 3} + \cdots + \frac{2}{d(d-1)} \right)\ket{11}\bra{11} \nonumber \\
        &\ + \left(\frac{2}{2 \cdot 3} + \frac{2}{3 \cdot 4} \cdots + \frac{2}{d(d-1)} \right) (\ket{22}\bra{22}+\ket{12}\bra{12} + \ket{21}\bra{21}) \nonumber \\
        &\ + \left(\frac{2}{3 \cdot 4} + \frac{2}{4 \cdot 5} \cdots + \frac{2}{d(d-1)} \right) (\ket{33}\bra{33}+\ket{13}\bra{13} + \ket{23}\bra{23} +\ket{31}\bra{31} + \ket{32}\bra{32}) \nonumber \\
        &\quad \vdots \nonumber \\
        &\ + \frac{2}{d(d-1)}\left(\ket{d-1,d-1}\bra{d-1,d-1} + \sum_{j = 1}^{d-2} \ket{j,d-1}\bra{j,d-1} + \ket{d-1,j}\bra{d-1,j} \right) \nonumber \\
        =&\ 2\sum_{b,b' = 1}^{d-1} \left(\sum_{c=\max(b,b')}^{d-1}\frac{1}{c(c+1)} \right) \ket{b,b'}\bra{b,b'} \nonumber\\
        =&\ 2\sum_{b,b' = 1}^{d-1} \left( \frac{1}{\max(b,b')} - \frac1{d} \right) \ket{b,b'}\bra{b,b'} \label{prop2eq3}
    \end{align}

    \noindent Plugging in \eqref{prop2eq2} and \eqref{prop2eq3} into $g'$ and recalling that $I = \sum_{b,b' = 1}^d \ket{b,b'}\bra{b,b'}$:

    \begin{align}
        g' &= \frac{1}{2}\left(\frac{d-1}{d}\right)\sum_{b,b' = 1}^d \ket{b,b'}\bra{b,b'} -\frac{1}{2}\sum_{b,b' = 1}^{d} \left( \frac{1}{\max(b,b')} - \frac1{d} \right) \ket{b,b'}\bra{b,b'} \nonumber\\
        &\ \ \ \ +\sum_{a=1}^{d-1}\frac{1}{2(a+1)}\left(\sum_{b=1}^a\ket{a+1, b} \bra{a+1,b} + \ket{b,a+1} \bra{b,a+1} \right) - \sum_{a=1}^{d-1}\frac{a}{2(a+1)}\ket{a+1,a+1}\bra{a+1,a+1}\nonumber\\
        &= \sum_{b,b' = 1}^d \frac{1}{2}\left(1-\frac{1}{\max(b,b')}\right)\ket{b,b'}\bra{b,b'} +\sum_{a=2}^{d}\left[\frac{1}{2a}\left(\sum_{b=1}^{a-1}\ket{a, b} \bra{a,b} + \ket{b,a} \bra{b,a} \right) - \frac{a-1}{2a}\ket{a,a}\bra{a,a} \right]\label{prop2eq4}
    \end{align}

    \noindent where the second summation re-indexes to starting at $a=2$. We can then rearrange \eqref{prop2eq4} into its diagonal and off-diagonal terms and simplify. 
    \begin{align}
    \begin{split}
        g' =&\sum_{b = 1}^d \frac{1}{2}\left(1-\frac{1}{b}\right)\ket{b,b}\bra{b,b} - \sum_{a=2}^{d}\frac{a-1}{2a}\ket{a,a}\bra{a,a} \\
        &\ \ \ \ + \sum_{b < a \in [d]} \frac{1}{2}\left(1-\frac{1}{a}\right)\left(\ket{a,b}\bra{a,b} + \ket{b,a}\bra{b,a}\right) + \sum_{b<a\in[d]} \frac{1}{2a}\left(\ket{a, b} \bra{a,b} + \ket{b,a} \bra{b,a} \right)
    \end{split} \nonumber\\
        =&\sum_{b = 1}^d \left(\frac{1}{2}\left(1-\frac{1}{b}\right) - \frac{1}{2}\left(1-\frac{1}{b}\right)\right)\ket{b,b}\bra{b,b}  + \sum_{b < a \in [d]} \frac{1}{2}\left(1-\frac{1}{a} + \frac{1}{a}\right)\left(\ket{a,b}\bra{a,b} + \ket{b,a}\bra{b,a}\right) \nonumber \\
        =& \frac{1}{2} \sum_{a \neq b \in [d]} \ket{a,b}\bra{a,b}\label{eq60}
    \end{align}

    Plugging this back into \eqref{prop2eq1} finishes the proof.
\end{proof}

In this proof, we considered a Hamiltonian, which we denoted by $g' := \frac{1}{2}\left(\frac{d-1}{d}\right) I - \frac{1}{4}\sum_{a = 1}^{d-1}\Lambda^d_a \otimes \Lambda^d_a$. This has special significance beyond this proof. Namely, we can give the following Proposition.

\begin{proposition}[The \textMCdProb Edge Interaction]
    The \textMCdProb edge interaction can be written in terms of the diagonal generalized Gell-Mann matrices in the following way. 

    \begin{equation}
        h^{\text{MC}_d} = \sum_{a\neq b \in [d]} \ket{a,b}\bra{a,b} = \left(\frac{d-1}{d}\right) I - \frac{1}{2}\sum_{a = 1}^{d-1}\Lambda^d_a \otimes \Lambda^d_a
    \end{equation}
\end{proposition}
\begin{proof}
    This follows from \eqref{eq60} in the proof of \cref{QMCk_edge_interaction_gell-mann_prop}.
\end{proof}

Additionally, \cref{prop_Psym_is_SUd_heis_plus_indent_term} follows as a result of \cref{QMCk_edge_interaction_gell-mann_prop}. 

\begin{proposition}[Restatement of \cref{prop_Psym_is_SUd_heis_plus_indent_term}]
    For \(P_{\text{sym}}\) being the orthogonal projector onto the symmetric subspace of \(\H_d^{\otimes 2}\), we have that \(\frac{1}{2}\left(\frac{d+1}{d}\right)I + h^{\text{Heis}_d} = P_{\text{sym}}\).
\end{proposition}
\begin{proof}
    Because the symmetric and the antisymmetric subspaces of \(\H_d^{\otimes 2}\) are orthogonal complements of each other, we have that \(P_{\text{sym}} = I - h\) (where \(h\) is defined in \cref{QMCd_edge_interaction_def}). Then, using \cref{QMCk_edge_interaction_gell-mann_prop}, we have that \(P_{\text{sym}} = \frac{1}{2}\left(\frac{d+1}{d}\right)I + \frac{1}{4}\sum_{a = 1}^{d-1}\Lambda^d_a \otimes \Lambda^d_a = \frac{1}{2}\left(\frac{d+1}{d}\right)I + h^{\text{Heis}_d}\).
\end{proof}

\section{Proofs for Bloch Vectors}\label{appendix_bloch_vecs}

We again consider the decomposition of density matrices into the generalized Gell-Mann (GGM) matrices given by
\begin{equation}\label{basic_bloch_vec_decomp_appendix}
    \rho = \frac{1}{d}I + \vec{b} \cdot \vec{\Lambda}
\end{equation}
We will also make use of the decomposition where pure states correspond to unit Bloch vectors. This is given by
\begin{equation}\label{rescaled1_bloch_vec_decomp_appendix}
    \rho = \frac{1}{d}I + \sqrt{\frac{d-1}{2d}} \vec{b} \cdot \vec{\Lambda}
\end{equation}
We note that both of these can be seen as shifting the origin to \(\rho^* = \frac{1}{d}I\), the maximally mixed state, and a change of basis, resulting in what we refer to as a Bloch vector. We note that the entries of \(\vec{b}\) are often called the mixture coordinates. 

We can then prove the propositions given in \cref{prod_state_section}. In doing this, it is useful to have a relation between the Hilbert-Schmidt inner product between two density matrices and the inner product between their Bloch vectors. 

\begin{proposition}\label{prop_tr_rho_rho_and_bloch_vec_inner_prod}
    For two density matrices \(\rho\) and \(\rho'\) with corresponding Bloch vectors \(\vec{b}\) and \(\vec{b}'\) as defined in \eqref{basic_bloch_vec_decomp_appendix}. We have the following relation.
    \begin{equation}
        \tr(\rho \rho') = \frac{1}{d} + 2 \braket{\vec{b}, \vec{b}'}
    \end{equation}
\end{proposition}
\begin{proof}
    \begin{align*}
        \tr(\rho \rho') &= \tr\left(\left(\frac{1}{d}I + \vec{b} \cdot \vec{\Lambda}\right)\left(\frac{1}{d}I + \vec{b}' \cdot \vec{\Lambda}\right)\right) \\
        &=  \tr\left(\frac{1}{d^2}I + \frac{1}{d}\vec{b} \cdot \vec{\Lambda} + \frac{1}{d} \vec{b}' \cdot \vec{\Lambda} + \left(\vec{b} \cdot \vec{\Lambda}\right) \left(\vec{b}' \cdot \vec{\Lambda}\right)\right) \\
        &= \frac{1}{d} + \sum_{a,b = 1}^{d^2-1} b_a b'_b \tr\left(\Lambda^a \Lambda^b\right) \\
        &= \frac{1}{d} + 2 \braket{\vec{b}, \vec{b}'}
    \end{align*}
\end{proof}

Moreover, this shows how this decomposition into a Bloch vector is also an isometry from the Hilbert-Schmidt space of hermitian matrices, \(\mathcal{HM}_d\) to \(\R^{d^2-1}\), the space that contains \(\Omega_d\). Namely,
\[D^2_2(\rho, \rho') := \frac{1}{2}\tr\left((\rho - \rho')^2\right) = \braket{\vec{b}-\vec{b}',\vec{b}-\vec{b}'} = \|\vec{b}-\vec{b}'\|^2_2\]

Next, we prove the following Propositions from \cref{prod_state_section}.

\begin{proposition}[Restatement of \cref{outsphere_prop}] Given a density matrix \(\rho\) and it's corresponding Bloch vector \(\vec{b}\) as defined in \eqref{basic_bloch_vec_decomp_appendix}, we always have that \(\|\vec{b}\| \leq \sqrt{\frac{d-1}{2d}}\). Moreover, if \(\rho\) is a pure state, i.e., \(\rho^2 = \rho\), then \(\|\vec{b}\| = \sqrt{\frac{d-1}{2d}}\).
\end{proposition}
\begin{proof}
    First, we recall the definitions of a density matrix, \(\rho\). That being (i) \(\tr(\rho) = 1\), (ii) \(\rho^* = \rho\), and (iii) \(\rho \succcurlyeq 0\). These properties imply that \(\tr(\rho^2) \leq 1\). This follows from the fact that \(\tr(\rho^2)\) is nothing but the sum of the eigenvalues squared and that
    \[1^2 = \tr(\rho)^2 = \left(\sum_{i = 1}^d \lambda_i\right)^2 = \sum_{i = 1}^d \lambda_i^2 + 2 \sum_{i < j} \lambda_i \lambda_j \geq \sum_{i = 1}^d \lambda_i^2\]
    The last inequality follows because the eigenvalues are non-negative. Then, using \cref{prop_tr_rho_rho_and_bloch_vec_inner_prod}, we get that
    \begin{align*}
        1 &\geq \tr(\rho^2) = \frac{1}{d} + 2 \|\vec{b}\|^2\\
        \|\vec{b}\|^2 &\leq \frac{d-1}{2d}
    \end{align*}
    Finally, we note that if \(\rho\) is a pure state, then we have that \(\tr(\rho^2) = \tr(\rho) = 1\), and thus we get an equality.
\end{proof}

Before proving anything about the insphere. We prove some other proposition that will be useful in the proof.

\begin{proposition} [Restatement of \cref{bloch_vec_angle_prop}, \cite{jakobczyk_geometry_2001,byrd_characterization_2003}]\label{bloch_vec_angle_prop_appendix}
    For two pure states, \(\rho\) and \(\rho'\), with Bloch vectors \(\vec{b}\) and \(\vec{b}'\) as defined in \eqref{rescaled1_bloch_vec_decomp_appendix}, we have that \(\braket{\vec{b},\vec{b}'} \geq -\frac{1}{d-1}\). Furthermore, these pure states are orthogonal, i.e., \(\tr(\rho \rho') = 0\), if and only if \(\braket{\vec{b},\vec{b}'} = -\frac{1}{d-1}\).
\end{proposition}
\begin{proof} This follows from \cref{prop_tr_rho_rho_and_bloch_vec_inner_prod} and that \(\tr(\rho \rho') \geq 0\) for any PSD matrices \(\rho\) and \(\rho'\).
    \begin{align*}
        0 &\leq \tr(\rho \rho') = \frac{1}{d} + 2 \left<\sqrt{\frac{d-1}{2d}} \vec{b}, \sqrt{\frac{d-1}{2d}} \vec{b}'\right> \\
        -\frac{1}{d}& \leq \frac{d-1}{d} \braket{\vec{b}, \vec{b}'} \\
        \braket{\vec{b}, \vec{b}'} &\geq -\frac{1}{d-1}
    \end{align*}
    And when \(\tr(\rho \rho') = 0\), we get an equality, not an inequality.
\end{proof}

\begin{proposition}[Restatement of \cref{prop_tr_rho2_vs_bloch_vec}]\label{prop_tr_rho2_vs_bloch_vec_appendix}
    For a density matrix, \(\rho\), with corresponding Bloch vector, \(\vec{b}\), as defined in \eqref{basic_bloch_vec_decomp_appendix}, we have that \(\tr(\rho^2) = \frac{1}{d-1}\) if and only if \(\|\vec{b}\| = \frac{1}{\sqrt{2d(d-1)}}\).
\end{proposition}

This follows as a direct result of \cref{prop_tr_rho_rho_and_bloch_vec_inner_prod}.

\begin{proposition}[Restatement of \cref{insphere_prop}]
Within \(\Omega_d\), defined using the decomposition of \eqref{basic_bloch_vec_decomp_appendix}, there is a maximal ball of radius \(\frac{1}{\sqrt{2d(d-1)}}\) that consists entirely of Bloch vectors that correspond to valid density matrices.
\end{proposition}
\begin{proof}
    Using \cref{prop_tr_rho2_vs_bloch_vec_appendix}, it is enough to show that given a unit trace, hermitian matrix \(\rho\) (i.e., \(\tr(\rho) = 1\) and \(\rho^* = \rho\)), we have that \(\rho \succcurlyeq 0\) if \(\tr(\rho^2) \leq \frac{1}{d-1}\). We prove this by contrapositive. That is, given a \(\rho \in \C^{d \times d}\) such that \(\tr(\rho) = 1\) and \(\rho^* = \rho\), we will prove that \(\lambda_\text{min}(\rho) < 0\) implies that \(\tr(\rho^2) > \frac{1}{d-1}\). Next, to show that the ball defined by \(\tr(\rho^2) \leq \frac{1}{d-1}\) is maximal, we show that there are vectors beyond its surface that do not correspond to valid density matrices.

    We start by using the fact that the trace is the sum of the eigenvalues. Furthermore, because \(\rho\) is hermitian, we know all the eigenvalues exist and are real. Namely, we define the ordering of eigenvalues \(\lambda_1 \geq \cdots \geq \lambda_r \geq 0 > \lambda_{r+1} \geq \cdots \geq \lambda_d\). That is, there are \(r < d\) non-negative eigenvalues. Then, let \(\epsilon = -\sum_{i=r+1}^d \lambda_i > 0\), which gives us that
    \[\sum_{i = 1}^{r} \lambda_i = 1 + \epsilon\]
    Next, we consider 
    \[\tr(\rho^2) = \sum_{i = 1}^r \lambda_i^2 + \sum_{i = r+1}^d \lambda_i^2 > \sum_{i = 1}^r \lambda_i^2\]
    We can then minimize this by setting \(\lambda_1 = \cdots = \lambda_r = \frac{1+\epsilon}{r}\), which gives us that
    \[\tr(\rho^2) > \sum_{i = 1}^r \lambda_i^2 \geq r \cdot \left(\frac{1+\epsilon}{r}\right)^2 = \frac{(1+\epsilon)^2}{r} \geq \frac{(1+\epsilon)^2}{d-1} > \frac{1}{d-1}\]

    Next, we show the maximality of the insphere by showing that there are vectors beyond the boundary of \(\Omega_d\) that do not correspond to valid density matrices.
    To do this, we consider an arbitrary orthonormal basis of states \(\{\ket{e_i}\}_{i \in [d]}\) and their density matrices \(\{\ket{e_i}\bra{e_i}\}_{i \in [d]}\). Let their respective Bloch vectors be given by \(\{\vec{b}_i\}_{i \in [d]}\), which are defined by the decomposition \eqref{rescaled1_bloch_vec_decomp_appendix}, and thus have that \(\|\vec{b}_i\| = 1\). Moreover, by \cref{bloch_vec_angle_prop_appendix}, we have that \(\braket{b_i, b_j} = \frac{-1}{d-1}\) for \(i \neq j\).
    
    Next, consider the sum all of these Bloch vectors but \(\vec{b}_1\), which gives the new Bloch vector, \(\vec{b}' = c\sum_{i = 2}^{k} \vec{b}_i\), with some constant \(c \geq 0\) that makes \(\rho' = \frac{1}{d}I + \vec{b}' \cdot \vec{\Lambda}\) a valid density matrix. Conveniently, we have that (for \(c=1\))
    \[\|\vec{b'}\|^2 = \sum_{i = 2}^{d} \|\vec{b}_i\|^2 + 2 \sum_{2 \leq i<j \leq d} \braket{b_i,b_j} = (d-1) - 2 \binom{d-1}{2} \frac{1}{d-1} = 1\]
    and so in general, \(\|\vec{b}'\| = c\). We finish the proof using the fact that \(\rho'\) must be PSD and thus \(\tr(\rho' A*A) \geq 0\) for all \(A \in \mathcal{L}(\H_d)\), and in particular \(\tr(\rho' \ket{e_1}\bra{e_1}) \geq 0\). Using \cref{prop_tr_rho_rho_and_bloch_vec_inner_prod}, we get the following.
    \[0 \leq \tr(\rho' \ket{e_1}\bra{e_1}) = \frac{1}{d} + 2 \braket{\vec{b}',\vec{b}_1} = \frac{1}{d} + c \sqrt{\frac{2(d-1)}{d}} \sum_{i = 2}^{d} \braket{\vec{b}_1, \vec{b}_i} = \frac{1}{d} - c \sqrt{\frac{2(d-1)}{d}}\]
    \[c \leq \frac{1}{\sqrt{2d(d-1)}}\]
\end{proof}

\section{Odd-Degree Terms and Proof of \texorpdfstring{\cref{degree_one_terms_lemma}}{Lemma \ref{degree_one_terms_lemma}}}\label{appendix_odd_deg_terms}

In this section, we discuss odd-degree terms in general and then prove \cref{degree_one_terms_lemma}. 

For \textQMCProb, it was proven that, without loss of generality, odd-degree terms could be ignored when optimizing over the density matrices \cite{hwang_unique_2022}. Namely, for every density matrix \(\rho\), there is another density matrix that has the same energy on the problem Hamiltonian with no odd-degree terms. However, this seems only to be the case for strictly even degree local Hamiltonian\footnote{By strictly even degree, we mean that the local Hamiltonians can be written using only even degree terms.} over qubits.
We prove by contradiction that odd-degree terms, in general, can not be ignored for strictly even degree local Hamiltonian over qudits. In particular, we show that the components of degree-three terms for \textQMCdProb (with \(d \geq 3\)), a \(2\)-LHP, can not all be zero even for a simple graph like the triangle.

\begin{lemma}[Odd-Degree Terms Can't, In General, Be Ignored]
    There does not exist an optimal solution, \(\rho \in \mathcal{D}(\H_d^{\otimes n})\), for the \textQMCdProb problem (\(d \geq 3\)), with interaction graph being the triangle, such that all of its components onto degree-three terms are zero.
\end{lemma}
\begin{proof}
    Consider a solution \(\rho\) such that \(\tr(\rho A) = 0\) for all \(A \in \mathcal{P}_d^n(3) \setminus \mathcal{P}_d^n(2)\). Then, for the problem Hamiltonian, \(H_\triangle := \frac{1}{3}\left(h_{12} + h_{23} + h_{13}\right)\), its square is given by \(H_\triangle^2 = \frac{1}{3}H_\triangle +\frac{1}{9}\left(\left\{h_{12}, h_{13}\right\} + \left\{h_{12}, h_{23}\right\} + \left\{h_{13}, h_{23}\right\}\right)\). Note, we use \(\{A, B\} := AB + BA\) to denote the anti-commutator. Next, we look at the term \(\{h_{uv}, h_{vw}\}\), to get the following.

    \begin{align*}
        \left\{h_{uv}, h_{vw}\right\} &= \left\{\frac{1}{2}\left(\frac{d-1}{d}\right) I - \frac{1}{4}\sum_{a = 1}^{d^2 - 1} \Lambda^a_u \Lambda^a_v, \frac{1}{2}\left(\frac{d-1}{d}\right) I - \frac{1}{4}\sum_{a = 1}^{d^2 - 1} \Lambda^a_v \Lambda^a_w\right\} \\
        &= \left(\frac{d-1}{d}\right)\left( \frac{1}{2}\left(\frac{d-1}{d}\right) I - \frac{1}{4}\sum_{a = 1}^{d^2 - 1} \Lambda^a_u \Lambda^a_v - \frac{1}{4}\sum_{a = 1}^{d^2 - 1} \Lambda^a_v \Lambda^a_w\right) + \left\{\frac{1}{4}\sum_{a = 1}^{d^2 - 1} \Lambda^a_u \Lambda^a_v, \frac{1}{4}\sum_{a = 1}^{d^2 - 1} \Lambda^a_v \Lambda^a_w\right\} \\
        &= \left(\frac{d-1}{d}\right)\left( h_{uv} + h_{vw} - \frac{1}{2}\left(\frac{d-1}{d}\right) I\right) + \frac{1}{4^2} \sum_{a,b = 1}^{d^2 - 1} \left\{\Lambda^a_u \Lambda^a_v, \Lambda^b_v \Lambda^b_w\right\} \\
        \phantom{\left\{h_{uv}, h_{vw}\right\}} &= \left(\frac{d-1}{d}\right)\left( h_{uv} + h_{vw} - \frac{1}{2}\left(\frac{d-1}{d}\right) I\right) + \frac{1}{4^2} \sum_{a,b = 1}^{d^2 - 1} \Lambda^a_u \left(\frac{4}{d}\delta_{ab}I + 2 \sum_{c = 1}^{d^2-1} d_{abc} \Lambda^c_v\right) \Lambda^b_w \\
        &= \left(\frac{d-1}{d}\right)\left( h_{uv} + h_{vw} - \frac{1}{2}\left(\frac{d-1}{d}\right) I\right) + \frac{1}{4d} \sum_{a = 1}^{d^2 - 1} \Lambda^a_u \Lambda^a_w + \mathcal{O}\\
        &= \left(\frac{d-1}{d}\right)\left( h_{uv} + h_{vw} - \frac{1}{2}\left(\frac{d-1}{d}\right) I\right) - \frac{1}{d}\left(h_{uw} - \frac{1}{2}\left(\frac{d-1}{d}\right) I\right) + \mathcal{O} \\
        &= \left(\frac{d-1}{d}\right)\left( h_{uv} + h_{vw}\right) - \frac{1}{d}h_{uw} - \frac{1}{2}\left(\frac{(d-1)(d-2)}{d^2}\right) I + \mathcal{O}
    \end{align*}

    Where \(\mathcal{O}\) is a linear combination of strictly degree-three terms.
    Plugging that into the equation for the square of \(H_\triangle\), we get the following.

    \begin{align*}
        H_\triangle^2 &= \frac{1}{3}H_\triangle +\frac{1}{9}\left(\left\{h_{12}, h_{13}\right\} + \left\{h_{12}, h_{23}\right\} + \left\{h_{13}, h_{23}\right\}\right) \\
        &= \frac{1}{3}H_\triangle + \frac{1}{3}\left(2 \left(\frac{d-1}{d}\right) H_\triangle - \frac{1}{d} H_\triangle - \frac{1}{2}\left(\frac{(d-1)(d-2)}{d^2}\right)I + \mathcal{O}\right) \\
        &= \left(\frac{d-1}{d}\right) H_\triangle - \frac{1}{6}\left(\frac{(d-1)(d-2)}{d^2}\right)I + \mathcal{O}
    \end{align*}
    Then, we consider the matrix \(M = I - H_\triangle\) and the energy of its square, for which we have that \(\tr(\rho M^2) \geq 0\) because \(\rho \succcurlyeq 0\). Additionally, we have that \(\tr(\rho \mathcal{O}) = 0\) by assumption.
    \begin{align*}
        0 \leq \tr(\rho M^2) &= \left(1 - \frac{1}{6}\left(\frac{(d-1)(d-2)}{d^2}\right)\right) - \left(2 - \left(\frac{d-1}{d}\right)\right) \tr(\Tilde{\rho}H_\triangle) + \tr(\rho \mathcal{O})\\
        &= \left(\frac{(d+1) (5 d-2)}{6 d^2}\right) - \left(\frac{d+1}{d}\right) \tr(\rho H_\triangle)
    \end{align*}
    \[\Rightarrow \tr(\rho H_\triangle) \leq \frac{5 d-2}{6 d}\]
    
    Which, for \(d \geq 3\), is strictly less than one. Thus, by \cref{lemma_d_clique_energy}, we know that \(\rho\) can not be the optimal solution. Moreover, \(\mathcal{O} \neq 0\) (for \(d \geq 3\), which is not the case for \(d=2\)).

\end{proof}

While odd-degree terms cannot be, in general, ignored, we observe that degree-one terms can. However, before we prove \cref{degree_one_terms_lemma}, we define the clock and shift matrices (also known as the Weyl operator basis \cite{bertlmann_bloch_2008}).

\begin{definition}[The Clock and Shift Matrices/Weyl Operator Basis \cite{bertlmann_bloch_2008,narnhofer_entanglement_2006}]
    The Weyl matrices are a matrix basis consisting of \(d^2\) traceless (with exception to the identity \(I = U_{00}\)), non-hermitian, unitary, pair-wise orthogonal matrices (namely, \(\tr(U^*_{ab} U_{cd}) = d \delta_{ac} \delta_{bd}\)). We can define them using the following two matrices. Note, it is more convenient to define them using a standard basis starting with \(\ket{0}\), i.e., \(\{\ket{i}\ |\ i \in [0,d-1]\}\). We also implicitly use modulo with the standard basis, i.e., \(\ket{k} := \ket{k\bmod d}\).
    \begin{itemize}
        \item The Clock Matrix
        \begin{equation}
            U_{10} := \sum_{k=0}^{d-1} \omega_d^{k} \ket{k} \bra{k}
        \end{equation}
        \item The Shift Matrix\footnote{Often, the shift matrix is defined as the transpose of this matrix.}
        \begin{equation}
            U_{01} := \sum_{k=0}^{d-1} \ket{k} \bra{k+1}
        \end{equation}
    \end{itemize}
    Where \(\omega_d:= e^{2\pi i/d}\) is the first primal \(d\)th root of unity. Then \(U_{nm} := U_{10}^n U_{01}^m\) for \(n,m \in [0,d-1]\). We can combine these two into a singular definition.
    \begin{equation}
        U_{nm} = \sum_{k=0}^{d-1} \omega^{kn} \ket{k} \bra{k+m},\ \ \ \ n,m \in [0,d-1]
    \end{equation}
\end{definition}
Of particular interest is their algebraic structure (which determines the Weyl operators up to unitary equivalence).
\begin{equation}
    U_{nm} U_{xy} = \omega_d^{nx}U_{n+x,m+y}
\end{equation}
\begin{equation}
    U^*_{nm} = \omega_d^{nm}U_{-n,-m}
\end{equation}

\begin{proposition}[Alternate definition for the \textQMCdProb edge interaction]
    The \textQMCdProb edge interaction (defined in \cref{QMCd_edge_interaction_def}) can be written in terms of the Weyl operator basis. This is left unproven.
    \begin{align}
        h &= \frac{1}{2} I - \frac{1}{2d} \sum_{n,m = 0}^{d-1} U_{nm}^* \otimes U_{nm}\\
        h &= \frac{1}{2}\left(\frac{d-1}{d}\right) I - \frac{1}{2d} \!\!\!\! \sum_{\substack{n,m = 0\\(n,m) \neq (0,0)}}^{d-1} \!\!\!\! U_{nm}^* \otimes U_{nm}
    \end{align}
\end{proposition}

We can now prove \cref{degree_one_terms_lemma}. Note, we continue to use as the standard basis, \(\{\ket{i}\ |\ i \in [0,d-1]\}\). We also use the general Gell-Mann matrices, which we write in this new standard basis (originally defined in \cref{gen_gell-mann_def}), which is still the same with the exception of the diagonal matrices, which are now defined as
\begin{equation}
    \Lambda^d_{a} :=  \sqrt{\frac{2}{(a+1)(a+2)}}   \left(\sum_{b=0}^{a} \ket{b}\bra{b}-(a+1)\ket{a + 1}\bra{a + 1}\right),\ \ \ \ 0 \leq a \leq d - 2  
\end{equation}
For brevity, let \(c_a := \sqrt{\frac{2}{(a+1)(a+2)}}\).

\begin{lemma}[Restatement of \cref{degree_one_terms_lemma}]
    Let \(\{h_\alpha\}_\alpha\) be a set of local Hamiltonians invariant under conjugation of local unitaries. For any (pseudo)-density matrix \(\rho = \frac{1}{2^n}\sum_{A \in \mathcal{P}_d^n} y(A) A\), there is another density matrix \(\rho' = \frac{1}{2^n}\sum_{A \in \mathcal{P}_d^n} y'(A) A\) such that that \(y'(A) = 0\) for all \(A \in \mathcal{P}_d^n(1) \setminus \{I\}\) being a degree one basis vector, which achieves the same expected energy, \(\E_\alpha \left[\tr(\rho h_\alpha)\right] = \E_\alpha \left[\tr(\rho' h_\alpha)\right]\).
\end{lemma}
\begin{proof}
    We say a Hamiltonian is invariant under conjugation by local unitaries if \(U^{\otimes n} H (U^*)^{\otimes n} = H\). Next, consider an arbitrary degree \(2t\) pseudo-density matrix \(\rho = \frac{1}{2^n}\sum_{A \in \mathcal{P}_d^n(2t)} y(A) A\) over \(n\) qudits and let \(\rho' := U^{\otimes n} \rho (U^*)^{\otimes n}\). We note that proving this lemma for pseudo-density matrices is a stronger claim than for true density matrices, as a true density matrix is the same as a degree-\(2n\) pseudo-density matrix (as per \eqref{pdensity_hierarchy}).
    
    We want to show that \(\rho'\) is also a valid degree \(2t\) pseudo-density matrix with the same objective. The latter follows easily as the \(h_\alpha\)'s are assumed to be invariant under conjugation by local unitaries.
    \[\tr(\rho' h_\alpha) = \tr(\rho (U^*)^{\otimes n} h_\alpha U^{\otimes n}) = \tr(\rho h_\alpha)\]
    And thus, \(\E_\alpha \left[\tr(\rho h_\alpha)\right] = \E_\alpha \left[\tr(\rho' h_\alpha)\right]\).
    Similarly, since the identity is such a matrix, we also have that
    \[\tr(\rho' I) = \tr(\rho (U^*)^{\otimes n} I U^{\otimes n}) = \tr(\rho I) = 1\]
    
    Lastly, to be a valid degree \(2t\) pseudo-density matrix, we have to show that \(\tr(\rho' A^* A) \geq 0\) for any matrix \(A\) of degree at most \(t\). To show this, first consider the following
    \begin{align*}
        \tr(\rho' A^* A) &= \tr(\rho (U^*)^{\otimes n} A^* A U^{\otimes n}) \\
        &= \tr(\rho (U^*)^{\otimes n} A^* U^{\otimes n} (U^*)^{\otimes n} A U^{\otimes n}) \\
        &= \tr(\rho \left((U^*)^{\otimes n} A U^{\otimes n}\right)^* (U^*)^{\otimes n} A U^{\otimes n})
    \end{align*}

    Then, we can observe that \(\omega((U^*)^{\otimes n} A U^{\otimes n}) \leq \omega(A)\). To show this, consider the decomposition of \(A\) in the basis of \(\mathcal{P}_d^n(t)\), \(A = \sum_{\tau \in \mathcal{P}_d^n(t)} \hat{A}(\tau) \tau\). Then, using the linearity of conjugation, we have that

    \[(U^*)^{\otimes n} A U^{\otimes n} = \sum_{\tau \in \mathcal{P}_d^n(t)} \hat{A}(\tau) (U^*)^{\otimes n} \tau U^{\otimes n}\]
    
    for which \(\omega((U^*)^{\otimes n} \tau U^{\otimes n}) \leq \omega(\tau)\) as all qudits that \(\tau\) acts as the indent, \((U^*)^{\otimes n} \tau U^{\otimes n})\) acts as the identity. Hence, \(\rho'\) is a valid degree \(2t\) pseudo-density matrix with the same objective value as \(\rho\).

    Finally, we get rid of degree-one terms. To do this, we can look at the convex combination of the conjugation over the clock and shift matrices. We note that any convex combination of degree \(2t\) pseudo-density matrices, which all give the same objective value, gives a new degree \(2t\) pseudo-density matrices with the same objective. So we consider 
    \[\rho'' = \E_{(n,m) \sim [0,d-1]^2}[U_{nm}^{\otimes n} \rho (U_{nm}^{\otimes n})^*]\]
    for which we isolate the degree 1 terms in \(\rho\), \(\tau \in \mathcal{P}_d^n(1)\setminus{I}\). We then prove this in two parts, for \(\tau = \Lambda^\pm_{ab}\) and \(\tau = \Lambda^d_a\). Specifically, we show that \(\E_{(n,m) \sim [0,d-1]^2}[U_{nm}^{\otimes n} \tau (U_{nm}^{\otimes n})^*] = 0\), which implies that \(\rho''\) has no degree 1 terms, proving the theorem.

    \begin{align*}
        U_{nm} \Lambda^+_{ab} U_{nm}^* &= \left(\sum_{k=0}^{d-1} \omega^{kn} \ket{k} \bra{k+m}\right) \Lambda^+_{ab} \left(\sum_{k=0}^{d-1} \omega^{-(k-m)n} \ket{k} \bra{k-m}\right) \\
        &= \left(\sum_{k=0}^{d-1} \omega^{(k-m)n} \ket{k-m} \bra{k}\right) \left(\ket{a}\bra{b} + \ket{b}\bra{a}\right) \left(\sum_{k=0}^{d-1} \omega^{-(k-m)n} \ket{k} \bra{k-m}\right) \\
        &= \omega^{(a-m)n-(b-m)n} \ket{a-m} \bra{b-m} + \omega^{(b-m)n-(a-m)n} \ket{b-m} \bra{a-m} \\
        &= \omega^{(a-b)n} \ket{a-m} \bra{b-m} + \omega^{-(a-b)n} \ket{b-m} \bra{a-m}
    \end{align*}
    Similarly,
    \[U_{nm} \Lambda^-_{ab} U_{nm}^* = -i \omega^{(a-b)n} \ket{a-m} \bra{b-m} + i\omega^{-(a-b)n} \ket{b-m} \bra{a-m}\]
    Then, we can observe that
    \begin{align*}
        \E_{(n,m) \sim [0,d-1]^2}[U_{nm}^{\otimes n} \Lambda^+_{ab} (U_{nm}^{\otimes n})^*] &= \frac{1}{d^2}\sum_{m=0}^{d-1}\sum_{n=0}^{d-1} \left(\omega^{(a-b)n} \ket{a-m} \bra{b-m} + \omega^{-(a-b)n} \ket{b-m} \bra{a-m}\right) \\
        &= \frac{1}{d^2}\sum_{m=0}^{d-1} \left(\ket{a-m} \bra{b-m} \sum_{n=0}^{d-1} \omega^{(a-b)n} + \ket{b-m} \bra{a-m} \sum_{n=0}^{d-1} \omega^{-(a-b)n}\right) \\
        &= \frac{1}{d^2}\sum_{m=0}^{d-1} 0\cdot \left(\ket{a-m} \bra{b-m} + \ket{b-m} \bra{a-m}\right) = 0
    \end{align*}
    And for the same reason,
    \[\E_{(n,m) \sim [0,d-1]^2}[U_{nm}^{\otimes n} \Lambda^-_{ab} (U_{nm}^{\otimes n})^*] = \frac{1}{d^2}\sum_{m=0}^{d-1} 0\cdot \left(-i\ket{a-m} \bra{b-m} + i\ket{b-m} \bra{a-m}\right) = 0\]

    Next, we consider \(\Lambda^d_a\).
    \begin{align*}
        U_{nm} \Lambda^d_{a} U_{nm}^* &= c_a \left(\sum_{k=0}^{d-1} \omega^{(k-m)n} \ket{k-m} \bra{k}\right) \left(\sum_{b=0}^{a} \ket{b}\bra{b}-(a+1)\ket{a + 1}\bra{a + 1}\right) \left(\sum_{k=0}^{d-1} \omega^{-(k-m)n} \ket{k} \bra{k-m}\right) \\
        &= c_a \left(\sum_{b=0}^{a} \omega^{(b-m)n - (b-m)n} \ket{b-m} \bra{b-m} - (a+1) \omega^{(a-m)n - (a-m)n} \ket{a+1-m} \bra{a+1-m}\right) \\
        &= c_a \left(\sum_{b=0}^{a} \ket{b-m} \bra{b-m} - (a+1) \ket{a+1-m} \bra{a+1-m}\right)
    \end{align*}
    Then,
    \begin{align*}
        \E_{(n,m) \sim [0,d-1]^2}[U_{nm}^{\otimes n} \Lambda^d_{a} (U_{nm}^{\otimes n})^*] &= \frac{c_a}{d^2}\sum_{n=0}^{d-1}\sum_{m=0}^{d-1} \left(\sum_{b=0}^{a} \ket{b-m} \bra{b-m} - (a+1) \ket{a+1-m} \bra{a+1-m}\right) \\
        &= \frac{c_a}{d^2}\sum_{n=0}^{d-1} \left(\sum_{b=0}^{a} \sum_{m=0}^{d-1} \ket{b-m} \bra{b-m} - (a+1) \sum_{m=0}^{d-1} \ket{a+1-m} \bra{a+1-m}\right) \\
        &= \frac{c_a}{d^2}\sum_{n=0}^{d-1} \left(\sum_{b=0}^{a} I - (a+1) I\right) \\
        &= \frac{c_a}{d^2}\sum_{n=0}^{d-1} \left((a+1) I - (a+1) I\right) = 0
    \end{align*}
\end{proof}

\section{Proof of \texorpdfstring{\cref{clique_val_prop}}{Proposition \ref{clique_val_prop}}}\label{appendix_proof_clique_val_prop}

\begin{proposition}[Restatement of \cref{clique_val_prop}]
    For \textQMCdProb with interaction graph, \(K_n\), being the unweighted complete graph on \(n\) vertices, the level-two ncSoS gets a value of \(\frac{(d-1) (d+n)}{2 d (n-1)}\).
\end{proposition}

Before we give the proof, we can observe that this is tight for \(n \geq d\) and \(d \mid n\). That is, this SDP gets the true optimal energy. The SDP overshooting for when \(d \nmid n\) is also true for the Frieze-Jerrum SDP.

\begin{proof}
    Let \(\{\ket{A}\ |\ A \in \mathcal{P}_d^n(1) \setminus \{I\}\}\) be a solution for the SDP, and let \(\Tilde{\rho}\) be the corresponding degree-two pseudo-density matrix as given by the correspondence, \(\braket{A|B} = \tr(\Tilde{\rho} AB)\). Also, let the labels of the vertices be given by \(V = [n]\). Recall that the level-2 ncSoS SDP is equivalent to \eqref{basicvecprog} without \(\rho_{uv}\) and constraints  \eqref{basicvecprog_loc_moment_and_M} through \eqref{basicvecprog_loc_moment_degree_1_terms2}.
    
    We first give an upper bound using the sum of squares proof technique. Namely, for \(r \geq 2\) non-identity basis matrices, \(\left(\Lambda^{a_i}_{u_i}\right)_{i \in [r]} \in \left(\mathcal{P}_d^n(1) \setminus \{I\}\right)^{\times r}\), that all pair-wise commute, we consider the degree-one matrix \(A = \sum_{i = 1}^r \Lambda^{a_i}_{u_i}\), and it's square.
    \[A^2 = \left(\sum_{i = 1}^r \Lambda^{a_i}_{u_i}\right)^2 = \sum_{i = 1}^r \left(\Lambda^{a_i}_{u_i}\right)^2 + 2 \sum_{i < j \in [r]}  \Lambda^{a_i}_{u_i} \Lambda^{a_j}_{u_j} = \frac{2r}{d} I + \mathcal{O} + 2 \sum_{i < j \in [r]}  \Lambda^{a_i}_{u_i} \Lambda^{a_j}_{u_j}\]
    Where \(\mathcal{O}\) is some linear combination of degree-one terms, which has \(\tr(\Tilde{\rho \mathcal{O}}) = 0\) by \cref{degree_one_terms_lemma}. Then we consider the following
    \begin{align*}
        0 &\leq \tr(\Tilde{\rho} A^2) = \frac{2r}{d} + 2 \tr\left(\Tilde{\rho} \sum_{i < j \in [r]}  \Lambda^{a_i}_{u_i} \Lambda^{a_j}_{u_j}\right) \\
        \Longrightarrow -\frac{r}{d} &\leq \tr\left(\Tilde{\rho} \sum_{i < j \in [r]}  \Lambda^{a_i}_{u_i} \Lambda^{a_j}_{u_j}\right) = \sum_{i < j \in [r]} \braket{\Lambda^{a_i}_{u_i}| \Lambda^{a_j}_{u_j}}
    \end{align*}
    Putting this into the expression for the value of the SDP, with \(r = n\), gives us that
    \begin{align*}
        \E_{(u,v) \sim E} \left[\frac{1}{2}\left(\frac{d-1}{d}\right) - \frac{1}{4} \sum_{a = 1}^{d^2-1}\braket{\Lambda^a_u | \Lambda^a_v}\right] &= \frac{1}{2}\left(\frac{d-1}{d}\right) - \frac{1}{4} \sum_{a = 1}^{d^2-1} \E_{(u,v) \sim E} \left[\braket{\Lambda^a_u | \Lambda^a_v}\right]\\
        &\leq \frac{1}{2}\left(\frac{d-1}{d}\right) - \frac{1}{4} \sum_{a = 1}^{d^2-1} \left(-\frac{2}{d(n-1)}\right) \\
        &= \frac{(d-1) (d+n)}{2 d (n-1)}
    \end{align*}

    We can then show that this is an equality by considering the following possible SDP matrix, \(M\), and in particular, its Cholesky decomposition, which gives SDP vectors that achieve the same value as above.

    \[M(\Lambda^a_u, \Lambda^b_v) = \begin{cases}
        \frac{2}{d} & \text{if } a = b \text { and } u = v \\
        -\frac{2}{d(n-1)} & \text{if } a = b \text { and } u \neq v \\
        0 & \text{otherwise}
    \end{cases}
    \]

    Finally, we can show that this matrix is PSD by observing that, up to some permutation of the rows and columns, it can be written as a block diagonal matrix, consisting of \(d^2-1\), blocks of size \(n \times n\), given by
    \begin{equation}
        \frac{2}{d} \begin{bmatrix}
            1 & -\frac{1}{n-1} & -\frac{1}{n-1} & \cdots & -\frac{1}{n-1} \\
            -\frac{1}{n-1} & 1 & -\frac{1}{n-1} & \cdots & -\frac{1}{n-1} \\
            -\frac{1}{n-1} & -\frac{1}{n-1} & 1 & \cdots & -\frac{1}{n-1} \\
            \vdots & \vdots & \vdots & \ddots & \vdots \\            
            -\frac{1}{n-1} & -\frac{1}{n-1} & -\frac{1}{n-1} & \cdots & 1 \\
        \end{bmatrix}
    \end{equation}
    which is nothing but the Gram matrix of the vectors in an \((n-1)\)-simplex, with the scalar \(\frac{2}{d}\). Meaning it and the matrix \(M\) are PSD and thus have a Cholesky decomposition. 
\end{proof}

\section{Proof of \texorpdfstring{\cref{lemma_d_clique_energy}}{Lemma \ref{lemma_d_clique_energy}}}\label{appendix_proof_lemma_d_clique_energy}

\begin{lemma}[Restatement of \cref{lemma_d_clique_energy}]
    The energy of \textQMCdProb with the complete graph on \(d\) vertices, \(K_d\), as it's interaction graph is 1.
\end{lemma}

\begin{proof}
    For this, we consider the completely antisymmetric state on \(d\), \(d\)-dimensional qudits, defined as follows.
    \begin{equation}\label{compl_antisym_state}
        \ket{\Psi} = \frac{1}{\sqrt{d!}}\sum_{\sigma \in S_d} \sgn(\sigma) \ket{\sigma(1)} \otimes \ket{\sigma(2)} \otimes \cdots \otimes \ket{\sigma(d)}
    \end{equation}

    We note that this is the unique antisymmetric state (up to a phase) in \(\H_d^{\otimes d}\). We can then write this state using the \(\binom{d}{2}\) dimensional antisymmetric subspace of \(\H_d^{\otimes 2}\) (which we note, the projector onto which is precisely the \textQMCdProb edge interaction from \cref{QMCd_edge_interaction_def}). That is, for \(\{i,j\} \subseteq [d]\), we define \(\ket{\psi_{ij}} = \frac{1}{\sqrt{2}}\left(\ket{i} \otimes \ket{j} - \ket{j} \otimes \ket{i}\right)\) to be the unique antisymmetric state (up to a phase) in \(\text{span}\{\ket{i}, \ket{j}\}^{\otimes 2}\). Then, the completely antisymmetric state can be written as follows.

    \begin{equation}\label{compl_antisym_state2}
        \ket{\Psi} = \frac{\sqrt{2}}{\sqrt{d!}}\sum_{\sigma \in A_d} \ket{\psi_{\sigma(1),\sigma(2)}} \otimes \ket{\sigma(3)} \otimes \cdots \otimes \ket{\sigma(d)}
    \end{equation}

    We prove this by considering the bijective function \(f: S_d \ra S_d, \sigma \mapsto (\sigma(1)\ \sigma(2)) \circ \sigma\) (\(f\) is a bijection because it is its own inverse: \(f(f(\sigma)) = (\sigma(2)\ \sigma(1)) \circ (\sigma(1)\ \sigma(2)) \circ \sigma = \sigma\) for all \(\sigma \in S_d\)). Next, we observe that when restricted to \(A_d\), the alternating group, we get \(\im(f |_{A_d}) = (1\ 2)A_d = \{\sigma \in S_d\ |\ \sgn(\sigma) = -1\}\), the co-set of odd permutations. Finally, we can do the following.
    
    \begin{align*}
        \ket{\Psi} &= \frac{1}{\sqrt{d!}}\sum_{\sigma \in S_d} \sgn(\sigma) \ket{\sigma(1)} \otimes \ket{\sigma(2)} \otimes \cdots \otimes \ket{\sigma(d)} \\
        &= \frac{1}{\sqrt{d!}}\left(\sum_{\sigma \in A_d} \ket{\sigma(1)} \otimes \ket{\sigma(2)} \otimes \cdots \otimes \ket{\sigma(d)} - \sum_{\sigma \in (1\ 2) A_d} \ket{\sigma(1)} \otimes \ket{\sigma(2)} \otimes \cdots \otimes \ket{\sigma(d)}\right) \\
        &= \frac{1}{\sqrt{d!}}\sum_{\sigma \in A_d} \bigg(\ket{\sigma(1)} \otimes \ket{\sigma(2)} \otimes \cdots \otimes \ket{\sigma(d)} - \ket{f(\sigma)(1)} \otimes \ket{f(\sigma)(2)} \otimes \cdots \otimes \ket{f(\sigma)(d)}\bigg) \\
        &= \frac{1}{\sqrt{d!}}\sum_{\sigma \in A_d} \bigg(\ket{\sigma(1)} \otimes \ket{\sigma(2)} \otimes \cdots \otimes \ket{\sigma(d)} - \ket{\sigma(2)} \otimes \ket{\sigma(1)} \otimes \cdots \otimes \ket{\sigma(d)}\bigg) \\
        &= \frac{1}{\sqrt{d!}}\sum_{\sigma \in A_d}  \bigg(\ket{\sigma(1)} \otimes \ket{\sigma(2)} - \ket{\sigma(2)} \otimes \ket{\sigma(1)}\bigg) \otimes \ket{\sigma(3)} \otimes \cdots \otimes \ket{\sigma(d)} \\
        &= \frac{\sqrt{2}}{\sqrt{d!}}\sum_{\sigma \in A_d}  \ket{\psi_{\sigma(1),\sigma(2)}} \otimes \ket{\sigma(3)} \otimes \cdots \otimes \ket{\sigma(d)} \\
    \end{align*}

    If we fix a \(\{u,v\} \subseteq [n]\), then we can do the same process to get the following by using the function \(f: S_d \ra S_d, \sigma \mapsto (\sigma(u)\ \sigma(v)) \circ \sigma\).
    \begin{equation}\label{compl_antisym_state3}
        \ket{\Psi} = \frac{\sqrt{2}}{\sqrt{d!}}\sum_{\sigma \in A_d} \ket{\psi_{\sigma(u),\sigma(v)}}_{u,v} \otimes \bigotimes_{k \in [d] \setminus \{u,v\}} \ket{\sigma(k)}_k
    \end{equation}
    
Finally, we can consider the energy of \(\ket{\Psi}\) on the \textQMCdProb Hamiltonian \(H_{K_d} = \frac{2}{d(d-1)} \sum_{u \neq v \in V} h_{uv}\). 

\begin{align*}
    \bra{\Psi} h_{uv} \ket{\Psi} &= \frac{2}{d!}\left(\sum_{\sigma \in A_d} \bra{\psi_{\sigma(u),\sigma(v)}}_{u,v} \otimes \bigotimes_{k \in [d] \setminus \{u,v\}} \bra{\sigma(k)}_k\right) h_{uv} \left(\sum_{\sigma \in A_d} \ket{\psi_{\sigma(u),\sigma(v)}}_{u,v} \otimes \bigotimes_{k \in [d] \setminus \{u,v\}} \ket{\sigma(k)}_k\right) \\
    &= \frac{2}{d!}\sum_{\sigma \in A_d} \bra{\psi_{\sigma(u),\sigma(v)}}_{u,v}\, h_{uv} \ket{\psi_{\sigma(u),\sigma(v)}}_{u,v}\\
    &=  \frac{2}{d!}\sum_{\sigma \in A_d} 1 = 1\\
\end{align*}
And thus, we have that 
\(\bra{\Psi}H_{K_d}\ket{\Psi} = \E_{(u,v) \sim E} \bra{\Psi} h_{uv} \ket{\Psi} = 1\)
\end{proof}

\end{document}